\documentclass[twoside,11pt]{article}

\usepackage{jmlr2e}

\usepackage{microtype}
\usepackage{graphicx}
\usepackage{booktabs}		

\usepackage{amsfonts}       
\usepackage{nicefrac}       
\usepackage{microtype}      
\usepackage{enumitem}

\let\proof\relax
\let\endproof\relax
\usepackage{amsmath,amssymb,amsthm}
\usepackage{tikz}
\usepackage[ruled,vlined]{algorithm2e}
\usetikzlibrary{patterns}
\usepackage{float}
\restylefloat{table}
\usepackage{subcaption}
\usepackage[labelformat=parens,labelsep=quad,skip=3pt]{caption}

\ShortHeadings{Approximate Submodular Functions}{Gupta, Pequito and Bogdan}
\firstpageno{1}

\begin{document}

\title{Approximate Submodular Functions and Performance Guarantees}

\author{\name Gaurav~Gupta \email ggaurav@usc.edu \\
       \addr Ming Hsieh Department of Electrical Engineering\\
       University of Southern California\\
       Los Angeles, CA 90089, USA
       \AND
       \name S\'ergio~Pequito \email goncas@rpi.edu \\
       \addr Department of Industrial and Systems Engineering\\
       Rensselaer Polytechnic Institute\\
       Troy, NY, USA
       \AND
       \name Paul~Bogdan \email pbogdan@usc.edu \\
       \addr Ming Hsieh Department of Electrical Engineering\\
       University of Southern California\\
       Los Angeles, CA 90089, USA
       }

\editor{}

\maketitle

\begin{abstract}
We consider the problem of maximizing non-negative non-decreasing set functions. Although most of the recent work focus on exploiting submodularity, it turns out that several objectives we encounter in practice are not submodular. Nonetheless, often we leverage the greedy algorithms used in submodular functions to determine a solution to the non-submodular functions. Hereafter, we propose to address the original problem by \emph{approximating} the non-submodular function and analyze the incurred error, as well as the performance trade-offs. To quantify the approximation error, we introduce a novel concept of $\delta$-approximation of a function, which we used to define the space of submodular functions that lie within an approximation error. We provide necessary conditions on the existence of such $\delta$-approximation functions, which might not be unique. Consequently, we characterize this subspace which we refer to as \emph{region of submodularity}. Furthermore, submodular functions are known to lead to different sub-optimality guarantees, so we generalize those dependencies upon a $\delta$-approximation into the notion of \emph{greedy curvature}. Finally, we used this latter notion to simplify some of the existing results and efficiently (i.e., linear complexity) determine tightened bounds on the sub-optimality guarantees using objective functions commonly used in practical setups and validate them using real data.
\end{abstract}

\begin{keywords}
  non-submodular functions, submodular approximation, region of submodularity, complex networks, curvature
\end{keywords}

\newtheorem{theorem}{Theorem}
\newtheorem{remark}{Remark}
\newtheorem{lemma}{Lemma}
\newtheorem{definition}{Definition}
\newtheorem{proposition}{Proposition}
\newtheorem{corollary}{Corollary}
\newtheorem{assumption}{Assumption}

\newenvironment{sproof}{%
  \renewcommand{\proofname}{Proof (sketch)}\proof}{\endproof}

\section{Introduction}

A multitude of problems in machine learning, control, game theory, economics can be modeled as discrete optimization \citep{das1,tropp,zhang,golovin,guillory, hoi}. Specifically, those captured by a maximization of set functions $f(S)$ subject to cardinality constraints \citep{xue,tzoumas2018TCNS,krause1,schnitzler,das}. The set functions chosen as objectives can have arbitrary structures that are problem-dependent. Nonetheless, in a quest to quantify sub-optimality in such discrete optimization problems, we often try to unveil structures (i.e., subclasses of functions with specific properties) that enable us to either develop efficient algorithms to determine the optimal solution, or to approximate the solutions when the problem is NP-hard yet with sub-optimality guarantees. Within the latter class, there was a surge for \emph{submodular} functions\footnote{A set function $f: 2^{\Omega}\rightarrow\mathbb{R}$ over a ground set $\Omega = \{1,2,\hdots,N\}$ is referred as submodular if and only if for all sets $S,T \subseteq \Omega$, $f(S\,\cup\,T) + f(S\,\cap\,T) \leq f(S) + f(T)$ \citep{bach}.}, whose approximate solution can be determined by a \emph{greedy algorithm} that determines the suboptimal solution by recursively adding the element which maximizes the objective. Also, this algorithm is known to achieve a constant performance bound \citep{feige, nemhauser, sviridenko, buchbinder}, which can be improved by using the concept of curvature in \citep{conforti, iyer}. 

Notwithstanding, some objectives of interest do no possess such properties \citep{bian}, e.g., Bayesian A-optimality, determinantal function, subset selection in $R^{2}$ objective \citep{das}, sparse approximation \citep{Das2,cevher}. Subsequently, it has been proposed to try to use the greedy algorithm, which empirically leads to a good performance. Is this a coincidence, or are there implicit features that justify some of the empirical performances obtained?

This is the quest we pursue in the present work. Specifically, we seek answers to the following questions:
\begin{itemize}
\item Is it possible to approximate an arbitrary set function
by a submodular one within a given error everywhere
(i.e., for all possible sets)?
\item Can we quantify the error incurred by such approximation?
\item For functions with the same approximation error, can we find some functions that are preferable, in the sense that if we perform suboptimal greedy algorithms, we are guaranteed to achieve a smaller sub-optimality gap?
\end{itemize}
In this work, we will see how the parameter \emph{submodularity ratio} can be used to interpret better the meaning of closeness to submodularity and how we can limit the error in the approximation. We will also see that \emph{curvature} of submodular function plays a vital role when it comes to the choice of approximation function offering same errors. The techniques developed will establish the relationship between non-submodular functions and submodular functions which may not be intuitive at first look but can be very useful.

\textbf{Main contributions}.

The present paper is motivated by the above questions, and in particular, the main contributions are as follows: 

\textbf{$\delta-$approximation}
We propose that any set function can be approximated as a $\delta-$approximate submodular function. This is used to capture the `divergence' between a given non-submodular function and a submodular function we use as an approximation.

\textbf{Region of submodularity} A novel notion of the region of submodularity is introduced which characterizes the class of submodular functions that are $\delta$-approximations of a given set function. Such notion can be used to better understand the meaning of closeness to be submodular.

\textbf{Performance bounds} We propose a novel definition for greedy curvature which unifies the existing one in \citep{conforti}. The proposed definition can be leveraged to simplify the proof of current theorems on curvature bounds \citep{conforti, bian}. Building upon it, we proved the performance bound for the greedy algorithm using $\delta-$approximate submodular functions. Remarkably, the computational complexity is only linear (in contrast with the ones available in the literature).

The performance bound will be lower than that of submodular function. We provide an intuition into the decrease in performance as a penalty paid to deviate from submodular functions in terms of approximation error.

\section{Preliminaries}

A discrete function $f:2^{N}\rightarrow \mathbb{R}^{+}$ defined over the ground set $\Omega = \{1,2,\hdots,N\}$ is monotone non-decreasing if for all $S \subseteq T\subseteq \Omega$, $f(S) \leq f(T)$. The marginal gain of an element $a \in \Omega$ with respect to a set $S\subseteq \Omega$ is defined as $f_{S}(a) = f(S\cup\{a\}) - f(S)$. 

Throughout this work, we are concerned about cardinality constrained maximization problem, i.e.,
\begin{equation}
\max\limits_{\vert S\vert \leq k,\, S\subseteq \Omega}f(S).
\label{eqn:maxProb}
\end{equation}
The objective function $f$ in the above discrete optimization problem will possibly be a non-submodular function. The optimal solution of this problem will be referred as $OPT$ such that $f(\Omega^{\ast}) = OPT$ and $\vert\Omega^{\ast}\vert = k$. Before stating the definition of approximate submodularity, we define the divergence between two discrete set functions.

\begin{definition}
The divergence $d$ between two non-decreasing set functions $f:2^{N}\rightarrow \mathbb{R}^{+}$ and $g:2^{N}\rightarrow \mathbb{R}^{+}$ defined over the same set $\Omega$ is denoted as
\begin{equation}
d(f, g) = \max\limits_{S\subseteq \Omega, \:a \in \Omega\setminus S}\left\vert \frac{f_{S}(a)}{g_{S}(a)} - 1\right\vert.
\end{equation}
\end{definition}

The $\delta-$approximate submodular function using the above definition of divergence is defined as follows.

\begin{definition}
A function $f$, possibly non-submodular, is referred as $\delta-$approximate submodular if there exists a non-decreasing submodular function $g:2^{N}\rightarrow \mathbb{R}^{+}$ such that $d(f,g) \leq \delta$.
\label{def:approxSubmod}
\end{definition}

The above definition can be re-written in its most useful form for the rest of the paper as
\begin{equation}
(1-\delta)\,g_{S}(a) \leq f_{S}(a) \leq (1+\delta)\,g_{S}(a),~\forall S\subseteq \Omega, a\in \Omega\setminus S.
\label{eqn:altDeltaApprox}
\end{equation}

Intuitively, any finite-valued set function can be looked as $\infty-$approximate submodular function, but we will restrict to the scenario where $\delta \leq 1$ in Definition\,\ref{def:approxSubmod}. It can be seen that $f$ is exactly submodular if and only if $\delta = 0$. For any non-negative set function $f$ the submodularity ratio was introduced as a parameter to measure closeness to submodularity \citep{das}. Furthermore, this notion can be generalized to the so-called generalized submodularity ratio as described in the next definition.

\begin{definition}[generalized submodularity ratio \citep{bian}] The submodularity ratio of a non-negative set function $f$ is given by
\begin{equation}
\gamma_{f} = \min\limits_{S,T \subseteq \Omega}\frac{\sum\nolimits_{t\in T\setminus S}f_{S}(t)}{f_{S}(T)}.
\label{eqn:genSubMod}
\end{equation}
\end{definition}

The $\gamma_{f}$ is such that $0\leq\gamma_{f}\leq 1$, and is equal to $1$ if and only if $f$ is submodular. Another useful parameter called curvature of submodular functions, which is a measure of deviation from modularity, was introduced in \citep{conforti} to write better performance bounds. The total curvature of a submodular function is defined as follows.
\begin{equation}
\alpha_{T} = 1 - \min\limits_{a\in\Omega}\frac{f_{\Omega\setminus a}(a)}{f(a)}.
\label{eqn:totalCurv}
\end{equation}

The $\alpha_{T}$ is such that $0\leq\alpha_{T}\leq 1$ and is equal to $0$ for modular/additive functions. Similar to the total curvature for submodular functions, we can write generalized curvature for any non-negative set function as

\begin{definition}[Generalized curvature \citep{bian}]
The curvature of a non-negative set function $f(.)$ can be written as
\begin{equation}
\alpha = 1 - \min\limits_{S,T \subseteq \Omega, s \in S\setminus T}\frac{f_{S\setminus\{s\}\,\cup\,T}(s)}{f_{S\setminus\{s\}}(s)}.
\end{equation}
\end{definition}
The generalized curvature is between $0$ and $1$ for the case of non-decreasing functions and for the particular case of $f$ being submodular it can be verified that $\alpha = \alpha_{T}$.

For the case of greedy algorithm, if $S_{G}$ denotes the ordered set solution of cardinality constrained maximization problem and $\Omega^{\ast}$ be the optimal solution such that $\vert S_{G}\vert = \vert\Omega^{\ast}\vert = k$ and $S_{G} = \{s_{1},s_{2},\hdots,s_{k}\}$, then we can define the greedy curvature in contrast to that in \citep{conforti} as follows.
\begin{definition}[Greedy curvature]
For a non-decreasing set function $f$ and $\Omega^{\ast}$ such that $\vert\Omega^{\ast}\vert = k$, the greedy curvature is defined as
\begin{eqnarray}
\alpha_{G} &=& 1 - \min\limits_{1\leq i\leq k}\left\{\min\limits_{a \in S_{G}\setminus (S_{G}^{i-1}\cup \Omega^{\ast})}\frac{f_{S_{G}^{i-1}\cup \Omega^{\ast}}(a)}{f_{S_{G}^{i-1}}(a)}, \vphantom{\min\limits_{\substack{a \in (S_{G}\cap\Omega^{\ast})\setminus S_{G}^{i-1} \\ i\leq j \leq k-1}}\frac{f_{S_{G}^{j-1}}(s_{j})}{f_{S_{G}^{i-1}}(a)}}  \min\limits_{\substack{a \in (S_{G}\cap\Omega^{\ast})\setminus S_{G}^{i-1} \\ i\leq j \leq k}}\frac{f_{S_{G}^{j-1}}(s_{j})}{f_{S_{G}^{i-1}}(a)} \right\},
\label{eqn:greedyCurv}
\end{eqnarray}
\label{def:greedyCurv}
\end{definition}
\noindent where $S_{G}^{i} = \{s_{1},s_{2},\hdots,s_{i}\}$ for $1\leq i\leq k$, $S_{G}^{0} = \phi$. 

The greedy curvature defined above is always less than or equal to greedy curvature defined in \citep{conforti}. The second term introduced in the above expression bounds the consecutive marginals of greedy selection. This kind of technique will play a key role in proving performance bounds of the greedy algorithm and will simplify the proofs to a great extent. It can be easily verified that $\alpha_{G}\leq \alpha$. The main results and interpretations are presented in the following section.

\section{Results}

We start by addressing the necessary conditions that a set function must satisfy to be $\delta-$approximate submodular. First, notice that if the function is submodular, then $\delta$ can be set to zero, and vice versa. On the other hand, if $\delta$ is different from zero, then we show that it cannot be arbitrary close. In other words, we obtain a necessary condition that establishes a \emph{`submodularity gap'}. Subsequently, we aim to characterize in further detail the properties of these functions, which lead us to introduce the notion of the region of submodularity (ROS). Within this region, there may exist different submodular functions, whose curvature differs and directly impact the optimality guarantees. Specifically, the lower the curvature is, the better the performance as shown in \citep{conforti}. Therefore, we aim to leverage these properties to show the performance of the greedy algorithm and obtain one of the main results of the paper, i.e., the improved constant optimality guarantees.

\subsection{Approximate submodularity}
\label{ssec:approxSubmod}

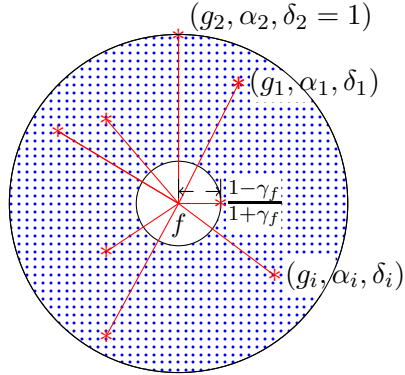
\begin{figure}
\centering
\begin{tikzpicture}[scale = 1.6]

\draw[pattern=dots,pattern color=blue] (0,0) circle (40pt);
\filldraw[fill=white] (0,0) circle (10pt);
\node[anchor=north] at (0,0) {\small$f$};

\draw[color=red] (0,0) -- (0,1.4);
\node[color=red] at (0,1.4) {$\ast$};
\draw[color=red] (0,0) -- (.35,0);
\node[color=red] at (0.35,0) {$\ast$};
\draw[color=red] (0,0) -- (0.5,1);
\node[color=red] at (0.5,1) {$\ast$};
\draw[color=red] (0,0) -- (-0.6,-0.4);
\node[color=red] at (-0.6,-0.4) {$\ast$};
\draw[color=red] (0,0) -- (0.8,-0.6);
\node[color=red] at (0.8,-0.6) {$\ast$};
\draw[color=red] (0,0) -- (-0.6,0.7);
\node[color=red] at (-0.6,0.7) {$\ast$};
\draw[color=red] (0,0) -- (-0.6,-1.1);
\node[color=red] at (-0.6,-1.1) {$\ast$};
\draw[color=red] (0,0) -- (-1,0.6);
\node[color=red] at (-1,0.6) {$\ast$};

\draw (0,0.125 - 0.08) -- (0,0.125 + 0.08);
\draw[->] (0.25/2,0.1) -- (0,0.1);

\draw[->] (0.45/2,0.1) -- (0.35,0.1) ;
\draw (0.35,0.125 - 0.08) -- (0.35,0.125 + 0.08);

\draw[color=red] (0,0) -- (-1,0.6);
\filldraw[fill=white,draw=white] (0.4,-.23) rectangle (0.85,0.23);
\node[anchor=west] at (0.3,0) {$\frac{1-\gamma_{f}}{1+\gamma_{f}}$};
\node[anchor=west] at (0,1.55) {$(g_{2},\alpha_{2},\delta_{2}=1)$};

\node[color=red] at (0.5,1) {$\ast$};
\filldraw[fill=white,draw=white] (0.6,0.85) rectangle (1.1,1.1);
\node[anchor=west] at (0.5,1) {$(g_{1},\alpha_{1},\delta_{1})$};

\filldraw[fill=white,draw=white] (0.9,-0.75) rectangle (1.3,-0.45);
\node[anchor=west] at (0.8,-0.6) {$(g_{i},\alpha_{i},\delta_{i})$};

\draw (0,0) circle (40pt);
\end{tikzpicture}
\caption{Region of submodularity (ROS) of $f$ as function of submodularity ratio $\gamma_{f}$. The minimum possible value of $\delta$ is $\frac{1-\gamma_{f}}{1+\gamma_{f}}$ and the maximum value is $1$. A tuple $(g_{i},\alpha_{i},\delta_{i})$ inside ROS represent the submodular function satisfying (\ref{eqn:altDeltaApprox}), its total curvature and corresponding $\delta_{i}$ respectively.}
\label{fig:ROS}
\end{figure}

\begin{lemma}
A non-submodular function $f$ with submodularity ratio $\gamma_{f}$ can be represented as $\delta-$approximate submodular to some submodular function $g$ only if $\delta \geq \frac{1-\gamma_{f}}{1+\gamma{f}}$.
\label{lemm:necessaryCond}
\end{lemma}
\begin{proof}
We will prove this by contradiction, let us assume that $\delta < \frac{1-\gamma_{f}}{1+\gamma_{f}}$ or $\frac{1+\delta}{1-\delta}<\frac{1}{\gamma_{f}}$. The submodularity ratio of any $g$ satisfying (\ref{eqn:altDeltaApprox}) can be written as
\begin{align*}
\gamma_{g} = \min\limits_{S,T \subseteq \Omega}\frac{\sum\nolimits_{t\in T\setminus S}g_{S}(t)}{g_{S}(T)} &\leq \frac{1+\delta}{1-\delta}\min\limits_{S,T \subseteq \Omega}\frac{\sum\nolimits_{t\in T\setminus S}f_{S}(t)}{f_{S}(T)} < \frac{1}{\gamma_{f}}\gamma_{f} = 1.
\end{align*}
where the first inequality is written using (\ref{eqn:altDeltaApprox}). Hence, $g$ cannot be submodular.
\end{proof}
The above result gives interesting insight into what do we mean by closeness of a function to being submodular using the submodularity ratio. The value of $\gamma_{f}$ limits the smallest possible value of $\delta$ and determines how close can be the function marginals to that of some submodular function. The \textbf{region of submodularity}, $\text{ROS}(f,\delta)$ is defined as.
\begin{equation}
\text{ROS}(f, \delta) = \left\{g~\vert~ \frac{1-\gamma_{f}}{1+\gamma_{f}}\leq d(f,g) \leq \delta\right\},
\end{equation}
\noindent which is the collection of $g$ for which a given function $f$ can be termed as $\delta-$approximate submodular. As shown in \figurename\,\ref{fig:ROS}, for a given $f$ the submodular functions in the shaded region can be used to describe it as an approximation. The closest value of $\delta$ is restricted by $\gamma_{f}$. It should be noted that multiple $g_{i}$ can be used to approximate the given $f$ as $\delta-$approximate submodular and hence, from the viewpoint of performance bound, we are interested in the $g$ with minimum total curvature for the given value of $\delta$. Let us denote $\alpha_{\delta}$ as the curvature of the selected $g$, i.e.,
\begin{equation}
\alpha_{\delta} = \min\limits_{g\in \text{ROS}(f,\delta)}\alpha_{T}(g).
\end{equation}

We will now state the result regarding performance of naive greedy selection for $\delta-$approximate submodular functions.

\subsection{Constant performance bound}
\label{ssec:constPerBnd}

It is worth to mention the result for the performance bound of a submodular function with total curvature $\alpha$.
\begin{theorem}[from \citep{conforti}]
For a non-negative non-decreasing submodular function $g$ with total curvature $\alpha$, if $OPT$ denotes the optimal value of $\max g(\mathcal{S})$ subject to $\vert\mathcal{S}\vert \leq k$, then the output of greedy algorithm, $\mathcal{S}_{G}$ satisfies $g(\mathcal{S}_{G}) \geq \frac{1}{\alpha}(1 - (1-\frac{\alpha}{k})^{k})OPT$.
\label{thm:curvature}
\end{theorem}
Also, for any set function (may or may not be $\delta-$approximate), the performance of greedy can be lower bounded as explained in the next result.
\begin{theorem}[from \citep{bian}]
Let $f$ be a non-negative non-decreasing set function with submodularity ratio $\gamma \in [0,1]$ and curvature $\alpha \in [0,1]$. The output of greedy algorithm satisfies $f(\mathcal{S}_{G}) \geq \frac{1}{\alpha}(1 - (1-\frac{\alpha\gamma}{k})^{k})OPT$.
\label{thm:submodCurv}
\end{theorem}

%
%

\begin{algorithm}
\SetKw{KwInitialize}{Initialize}
\SetArgSty{normal}
%
\KwOut{$S_{G}$}

\KwInitialize{$S_{G} = \phi$\;}

\Repeat{$\vert S_{G}\vert = k$}{

$s^{\ast} = \text{arg}\max\limits_{s \in \Omega\setminus S_{G}}f_{S_{G}}(s)$\;

$S_{G} \leftarrow S_{G} \cup s^{\ast}$\;

}
\caption{Greedy Selection Algorithm}
\label{alg:greedy_alg}
\end{algorithm}
%
For $\delta-$approximate submodular functions it can be shown that the greedy selection algorithm offers constant performance bound as captured in the next result.
\begin{theorem}
For a given $\delta-$approximate submodular function with submodularity ratio $\gamma_{f}$ and $\alpha_{\delta}$ such that $\delta \geq \frac{1-\gamma_{f}}{1+\gamma_{f}}$, the greedy algorithm has guaranteed constant performance.
\label{thm:deltaApproxGreedy}
\begin{align}
&f(S_{G}) \geq \frac{1}{\alpha_{\delta}\frac{1-\delta}{1+\delta} + \frac{2\delta}{1+\delta}}\left(1 - \left(1 - \frac{1}{k}\left(\alpha_{\delta}\frac{1-\delta}{1+\delta}+\frac{2\delta}{1+\delta}\right)\left(\frac{1-\delta}{1+\delta}\right)\right)^{k}\right)OPT.
\end{align}
\end{theorem}
\begin{sproof}
The detailed proof is provided in the Appendix. The main idea of the proof is to track the deviation of the optimal solution with the greedy algorithm at each step. The $\delta-$submodularity property can be exploited to bound the given function marginals in terms of some submodular function which enable us to apply the properties of submodularity. The greedy curvature defined in (\ref{eqn:greedyCurv}) is used to bound the successive marginals occurring at any $i$th step. Formally, let us denote the output of the greedy algorithm at the $i$th stage as $S_{G}^{i}$. The term $f(\Omega^{\ast} \cup S_{G}^{i})$ can be expanded in two different ways as follows:
\begin{eqnarray*}
f(\Omega^{\ast} \cup S_{G}^{i}) \geq f(\Omega^{\ast}) + (1-\alpha_{G})\sum\limits_{s_{j}\in S_{G}^{i}\setminus \Omega^{\ast}}f_{S_{G}^{j-1}}(s_{j}),
\end{eqnarray*}
and
\begin{equation}
f(\Omega^{\ast}\cup S_{G}^{i}) = f(S_{G}^{i}) + f_{S_{G}^{i}}(\Omega^{\ast}).
\label{eqn:midEqn2}
\end{equation}
The above equations can be combined and re-written after using the property of $\delta-$submodularity from (\ref{eqn:altDeltaApprox}) as
\begin{eqnarray*}
f(\Omega^{\ast}) &\leq& \alpha_{G}f(S_{G}^{i}) + (1-\alpha_{G})\sum\limits_{s_{j}\in S_{G}^{i}\cap\Omega^{\ast}}f_{S_{G}^{j-1}}(s_{j}) + \frac{1+\delta}{1-\delta}\sum\limits_{\omega\in\Omega^{\ast}\setminus S_{G}^{i}}f_{S_{G}^{i}}(\omega).
\label{eqn:midEqn3}
\end{eqnarray*}
Besides, the greedy algorithm at the $i$th step would select the index $s_{i+1}$ as
\begin{equation}
f_{S_{G}^{i}}(s_{i+1}) = \max\limits_{a \in \Omega\backslash S_{G}^{i}}f_{S_{G}^{i}}(a).
\label{eqn:greedyNextStep}
\end{equation}
Therefore, the equation (\ref{eqn:midEqn3}) can be re-written by upper bounding the last term using (\ref{eqn:greedyNextStep}) as
\begin{eqnarray*}
f(\Omega^{\ast})&\leq& \alpha_{G}f(S_{G}^{i}) + \sum\limits_{s_{j}\in S_{G}^{i}\cap\Omega^{\ast}}\left\{(1 - \alpha_{G})f_{S_{G}^{j-1}}(s_{j}) - f_{S_{G}^{i}}(s_{i+1})\right\} + k\frac{1+\delta}{1-\delta}f_{S_{G}^{i}}(s_{i+1}).
\end{eqnarray*}
The middle term in the above equation can be eliminated by exploiting the second term in the definition of greedy curvature in (\ref{eqn:greedyCurv}). The remaining inequality then reduces to the form of $\lambda u_{i+1} \geq c - \sum\nolimits_{j=1}^{i}u_{j}$, which has the solution of the form $\sum\nolimits_{i=1}^{k}u_{i} \geq c(1 - (1 - \frac{1}{\lambda})^{k})$ obtained using mathematical induction. Finally, we can write
\begin{align*}
f(S_{G}) &= \sum\limits_{i=1}^{k}f_{S_{G}^{i-1}}(s_{i})\\
&\geq \frac{1}{\frac{2\delta}{1+\delta} + \frac{1-\delta}{1+\delta}\alpha_{\delta}}
\left(1 - \left(1 - \frac{1}{k}\left(\frac{2\delta}{1+\delta} + \frac{1-\delta}{1+\delta}\alpha_{\delta}\right)\left(\frac{1-\delta}{1+\delta}\right)\right)^{k}\right)OPT.
\label{eqn:finalPerfBnd}
\end{align*}
\end{sproof}
\begin{table*}[t]
\centering
\begin{tabular}{|c|c|c|}
\hline
Approximation & Feasibility & Performance bound\\
\hline
$(1-\delta)g_{S}(a)\leq f_{S}(a) \leq (1+\delta)g_{S}(a)$ & $\delta \geq \frac{1-\gamma_{f}}{1+\gamma_{f}}$ & $\frac{1}{\alpha_{\delta}\frac{1-\delta}{1+\delta} + \frac{2\delta}{1+\delta}}\left(1 - \text{exp}\{-(\alpha_{\delta}\frac{1-\delta}{1+\delta} + \frac{2\delta}{1+\delta})\frac{1-\delta}{1+\delta}\}\right)$\\
\hline
$\delta_{l}g_{S}(a)\leq f_{S}(a) \leq \delta_{u}g_{S}(a)$ & $\delta_{l} \leq \delta_{u}\gamma_{f}$ & $\frac{1}{1 - \frac{\delta_{l}}{\delta_{u}}(1 - \alpha_{\delta})}\left(1 - \text{exp}\{-(1 - \frac{\delta_{l}}{\delta_{u}}(1 - \alpha_{\delta}))\frac{\delta_{l}}{\delta_{u}}\}\right)$\\
\hline
$g_{S}(a)\leq f_{S}(a)\leq \delta g_{S}(a)$ & $\delta \geq \frac{1}{\gamma_{f}}$ & $\frac{1}{1 - \frac{1 - \alpha_{\delta}}{\delta}}\left(1 - \text{exp}\{-(1 - \frac{1 - \alpha_{\delta}}{\delta})\frac{1}{\delta}\}\right)$\\
\hline
\end{tabular}
\caption{Feasibility and performance bound for different approximation structures.}
\label{tab:approxTab}
\end{table*}
For $\delta=0$, the above result is then same as Theorem\,\ref{thm:curvature} because  $f$ would be submodular. For $\delta \geq \frac{1-\gamma_{f}}{1+\gamma_{f}}$, the performance guarantee reduces due to divergence from submodularity property. While $\alpha_{\delta}$ corresponds to the total curvature of a submodular function used to approximate the given function $f$, the term $\frac{2\delta}{1+\delta}$ can be looked upon as penalty paid for deviating from this submodular function.

It is worth to mention that the proofs of the Theorem\,\ref{thm:curvature} and Theorem\,\ref{thm:submodCurv} can be simplified to a large extent if we use new strategies employing our proposed unified definition of the greedy curvature in (\ref{eqn:greedyCurv}). The direction of constructing the Linear Programs (LP) introduced in \citep{conforti} can then be skipped. Consequently, we provide an alternate version of the proof of Theorem\,\ref{thm:submodCurv} in the appendix for reference.

The definition of approximate submodularity in (\ref{eqn:altDeltaApprox}) is sufficient for identifying such functions but it often happens that the upper and lower bounds may not be symmetric as we will see in the next section. However, instead of making the bounds loose to bring regularity, such asymmetric behavior can be leveraged to get tighter bounds for the performance of greedy algorithm from Theorem\,\ref{thm:deltaApproxGreedy}. The feasibility conditions from Lemma\,\ref{lemm:necessaryCond} will also change accordingly. The variations of such asymmetric behavior is summarized in Table\,\ref{tab:approxTab}. Next, we identify some important examples that can be recognized as approximate submodular functions according to the Definition\,\ref{def:approxSubmod}.

\section{Applications}
\label{sec:Appl}

In this section, we identify some critical functions which are not submodular and appear a lot in the applications like sensor selection, sparse learning, etc. These functions are known to perform well with greedy search techniques, and we will establish that being $\delta-$approximate submodular helps in improving their existing performance bounds as well as developing some closeness guarantee to being submodular.

\subsection{Trace of inverse Gramian}

Let us consider a matrix $W_{S}$, which for the rest of the paper is denoted in its simplest form as $W_{S} = \Lambda_{0} + \sum\nolimits_{s\in S}x_{s} x_{s}^{T}$, where $x_{i}\in\mathbb{R}^{n}$ is taken from the data matrix $X\in\mathbb{R}^{n\times N}, X = [x_{1},x_{2},\hdots,x_{N}]$. The ordered eigenvalues of $W_{S}$ are denoted as $\lambda_{n}\geq \hdots \geq \lambda_{1}$.

The negative trace of the covariance matrix inverse is used as criteria for Bayesian-A optimality in \citep{krause1}. The problem is stated in \citep{bian} as variance reduction of the Bayesian estimated unknown. For the sake of completeness, we will define the problem here as follows. Given some set of observations ${\bf y}\in \mathbb{R}^{N}$ and the linear model $y = X^{T}\theta + w$ which relates the parameter $\theta$ with observations in the presence of Gaussian noise $w\sim\mathcal{N}(0,I_{N})$. The problem of sensor selection can be stated as minimizing the conditional variance of $\theta$ with a given sensor budget. If $S\subseteq \{1,2,\hdots,N\}$ denotes the set of selected sensors, then we have $y_{S} = X_{S}^{T}\theta + w$, where $y_{S}$ are the set of observations indexed according to the elements in $S$ and $X_{S}$ has columns taken from $X$ accordingly. With the Gaussian prior assumption for $\theta\sim\mathcal{N}(0,\Lambda_{0}^{-1})$ with $\Lambda_{0} = \beta^{2}I_{N}$, the conditional covariance of $\theta$ given $y_{S}$ can be written as $\Sigma_{\theta\vert y_{S}} = (\Lambda_{0} + X_{S}X_{S}^{T})^{-1}$. The objective function is defined as
\begin{equation}
f(S) = \text{tr}(\Lambda_{0}) - \text{tr}(\Lambda_{0} + \sum\nolimits_{s\in S}x_{s} x_{s}^{T})^{-1}.
\label{eqn:negTraceInv}
\end{equation}
The negative trace of the inverse of matrix, $W_{S}$, is a non-submodular function. We will see that it can be labeled as $\delta-$approximate submodular to a known submodular function, log determinant of $W_{S}$, using the following result.

\begin{proposition}
The negative trace of matrix inverse, $f(S) = -\text{tr}(W_{S}^{-1})$ is approximately submodular to the submodular function, log determinant of the matrix, $g(S) = \text{log}\,\text{det}(W_{S})$ with the upper and lower bounds as $\delta_{u} = \frac{1}{\lambda_{1}(W_{\phi})}$ and $\delta_{l} = \frac{1}{\lambda_{n}(W_{\Omega})}$, respectively.
\label{prop:negTrace}
\end{proposition}
Interestingly, this result establishes the relationship between two functions, negative trace of matrix inverse (non-submodular) and log determinant of the matrix (submodular), which may not look intuitive at first glance. We will see the usefulness of this kind of association in the Experiments section.

The matrix of the form of $W_{S}$ can be extended to Gramian which has applications in the domain of complex networks. The controllability of complex networks plays a fundamental role in network science, cyber-physical systems, and systems biology. The steering of networks has applications ranging from drug design for cancer networks to electrical grids. It has been realized in \citep{Pasqualetti} that negative trace of the observability Gramian inverse can be used to quantize the controllability of the networks subject to some sensors.

\subsection{Minimum Eigenvalue}

While controllability helps to steer the networks, the observability criterion of complex networks is the problem of selecting sensors and an estimation method such that the networks can be reconstructed from the measurements collected by the selected sensors. It has been stated in \citep{Pasqualetti} that minimum eigenvalue of the so-called Observability Gramian can be used to quantize the energy associated with a network state. The maximization of minimum eigenvalue also serves as criteria for matrix inversion in the presence of numerical errors for sparse matrices.

The minimum eigenvalue of the Gramian
\begin{equation}
f(S) = \lambda_{1}(W_{S}),
\label{eqn:minEigGramian}
\end{equation}
is non-submodular function \citep{Summers2016OnSA} but we now show that it can be modeled as a $\delta-$approximate submodular function.

\begin{proposition}
The minimum eigenvalue of Gramian, $f(S) = \lambda_{1}(W_{S})$ is approximately submodular to the modular function, trace of Gramian, $g_{1}(S) = \text{tr}(W_{S})$ with the upper and lower bounds as $\delta_{u}$ and $\delta_{l}$, respectively, where 
\begin{eqnarray*}
\delta_{u} &=& 1 - \frac{n-1}{n}\min\limits_{\omega \in\Omega}\frac{\lambda_{1}(W_{\omega})}{\lambda_{n}(W_{\omega})},\nonumber\\
\delta_{l}&=& \frac{1}{n}\min\limits_{\omega\in\Omega}\frac{\lambda_{1}(W_{\omega})}{\lambda_{n}(W_{\omega})}.
\end{eqnarray*}
\label{prop:minEig1}
\end{proposition}
The above result represents the non-submodular function $f$ as an approximation to a modular function. Since $g_{1}$ is modular, therefore its curvature is $\alpha = 0$. This property can be leveraged in Table\,\ref{tab:approxTab} to get tight bounds when $n$ is not large. Next, we show that there exists another function in the ROS of $f$ through the following result.

\begin{proposition}
The minimum eigenvalue of Gramian, $f(S) = \lambda_{1}(W_{S})$ is approximately submodular to the maximum eigenvalue of Gramian, $g_{2}(S) = \lambda_{n}(W_{S})$, which is submodular with the upper and lower bounds as $\delta_{u} = \max\limits_{\omega\in\Omega}\frac{\lambda_{n}(W_{\omega})}{\lambda_{1}(W_{\omega})}$ and $\delta_{l}= \min\limits_{\omega\in\Omega}\frac{\lambda_{1}(W_{\omega})}{\lambda_{n}(W_{\omega})}$, respectively.
\label{prop:minEig2}
\end{proposition}

The proofs of the Proposition\,\ref{prop:negTrace}, \ref{prop:minEig1} and \ref{prop:minEig2} are provided in the appendix. With some of these theoretical applications we will see their practical use in the following section.

\section{Experiments}

We apply some of the identified non-submodular functions in the previous section to real-world data and observe the performance of greedy selection. Specifically, first, we establish using simulations that the bounds computed using $\delta-$approximate submodularity concept widely improve the state-of-the-art bounds. Next, we apply the non-submodular function to the problem of finding sensors for the structured data using augmented linear model in the context of real-world electroencephalogram (EEG) dataset.

\subsection{Tightness analysis of the performance bounds}

\begin{figure}
\centering
\includegraphics*[viewport=0 0 600 460, width = 3in, height = 2.5in]{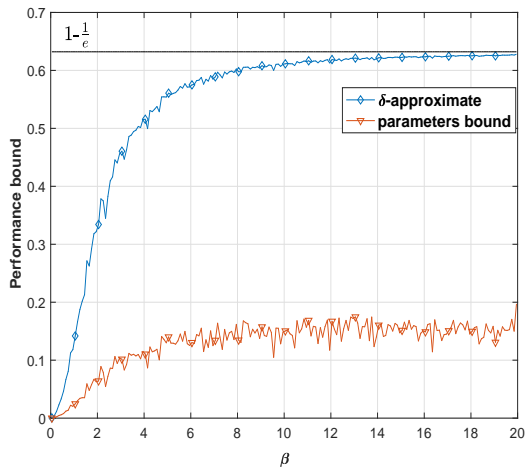}
\caption{Comparison of performance bound for non-submodular function $f(S) = -\text{tr}(W_{S}^{-1})$ using $\delta-$approximate method and using bound of parameters $(\alpha, \gamma)$ from \citep{bian}.}
\label{fig:perBnd}
\end{figure}

The performance bound presented in \citep{bian} requires a combinatorial exhaustive search for the computation of parameters, making it extremely difficult to realize in practical scenarios. To remedy this, the authors have presented the bounds for parameters (curvature and submodularity ratio). The performance bound for greedy algorithm presented in this work in Theorem\,\ref{thm:deltaApproxGreedy} requires only linear search to compute the curvature of the submodular function, $\alpha_{\delta}$ from (\ref{eqn:totalCurv}). We make a comparison between the presented bounds and the ones in  \citep{bian} through simulation of negative trace of $W_{S}$ inverse in (\ref{eqn:negTraceInv}). To simulate the matrix $W_{S}$, the entries of the data matrix ${\bf X}\in\mathbb{R}^{10\times 30}$ are generated from Gaussian distribution $\mathcal{N}(0,1)$ and each column is normalized, i.e. $\vert\vert{\bf x}_{i}\vert\vert_{2} = 1$. The matrix $\Lambda_{0}$ in $W_{S}$ is taken as $\beta^{2}I$. \figurename\,\ref{fig:perBnd} shows the performance bound for different values of $\beta$. It can be observed that there is a huge gap between the performance bounds using $\delta_{u}$ and $\delta_{l}$ from Proposition\,\ref{prop:negTrace}, Table\,\ref{tab:approxTab} and the one using bounds of parameters from \citep{bian}.

An interesting observation can be made when $\beta$ grows large which can be explained using the current theory of $\delta-$approximation. In the limit of large $\beta$, the ratio of bounds $\delta_{u}/\delta_{l}\rightarrow 1$, also  the identified submodular function, i.e., log determinant of Gramian becomes constant and, hence $\alpha_{\delta} \rightarrow 1$. Substituting these values in Table\,\ref{tab:approxTab}, the performance bound goes to $1-1/e$ which matches that of submodular function.

\subsection{Efficient sensor selection for accurate brain activity mining}

We now present some experiments for a real-world application setting collected by the BCI$2000$ system with a sampling rate of $160$Hz \citep{schalk}. Individually, we consider $64-$channel EEG data set which records the brain activity of $10$ subjects $(\text{S}001-\text{S}010)$ when they are performing motor and
imagery tasks. Each subject completed a set of $4$ tasks while interacting with a target appearing on the screen to imagine motor movements \citep{Goldbergere215}. Briefly, Task $1$ was open/closure of left or right fist as the target appears on the screen, Task $2$ was imagining the open/closure of left or right fist as the target appears. The Task $3$ was open/closure of both fist/feet and Task $4$ was imagining open/closure of both fist/feet as the target appears on the screen. The brain activity recorded through EEG sensors during this process is mathematically modeled using the augmented linear model as described in \citep{xue, gauravACC2018}.

\begin{figure}[!t]
\centering
\begin{subfigure}[b]{0.5\linewidth}
\centering
\includegraphics*[viewport=0 0 500 425, width = 3in, height = 2.8in]{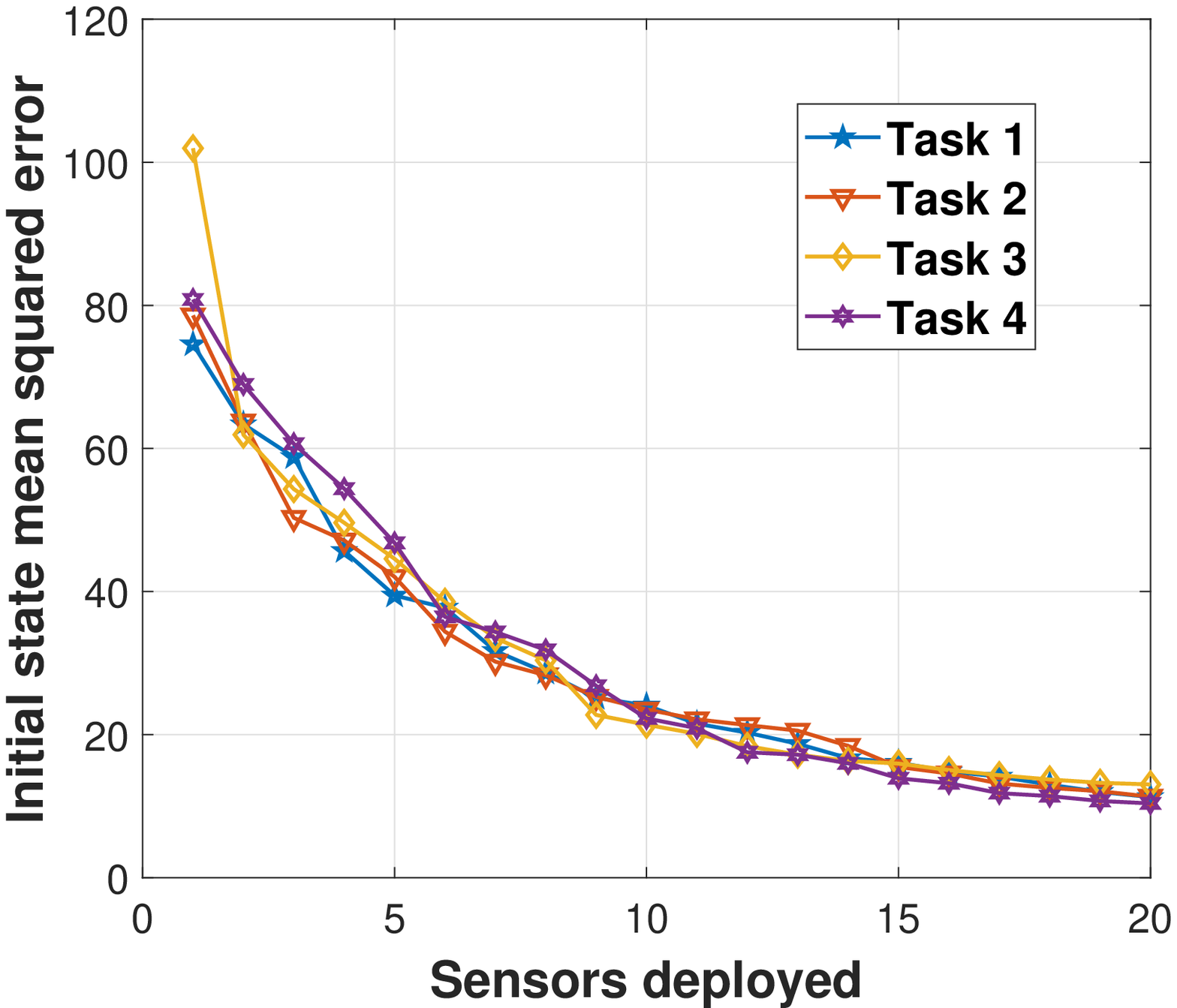}
\caption{}
\end{subfigure}\hfill
\begin{subfigure}[b]{0.5\linewidth}
\centering
\includegraphics*[viewport=0 0 510 425, width = 3in, height = 2.8in]{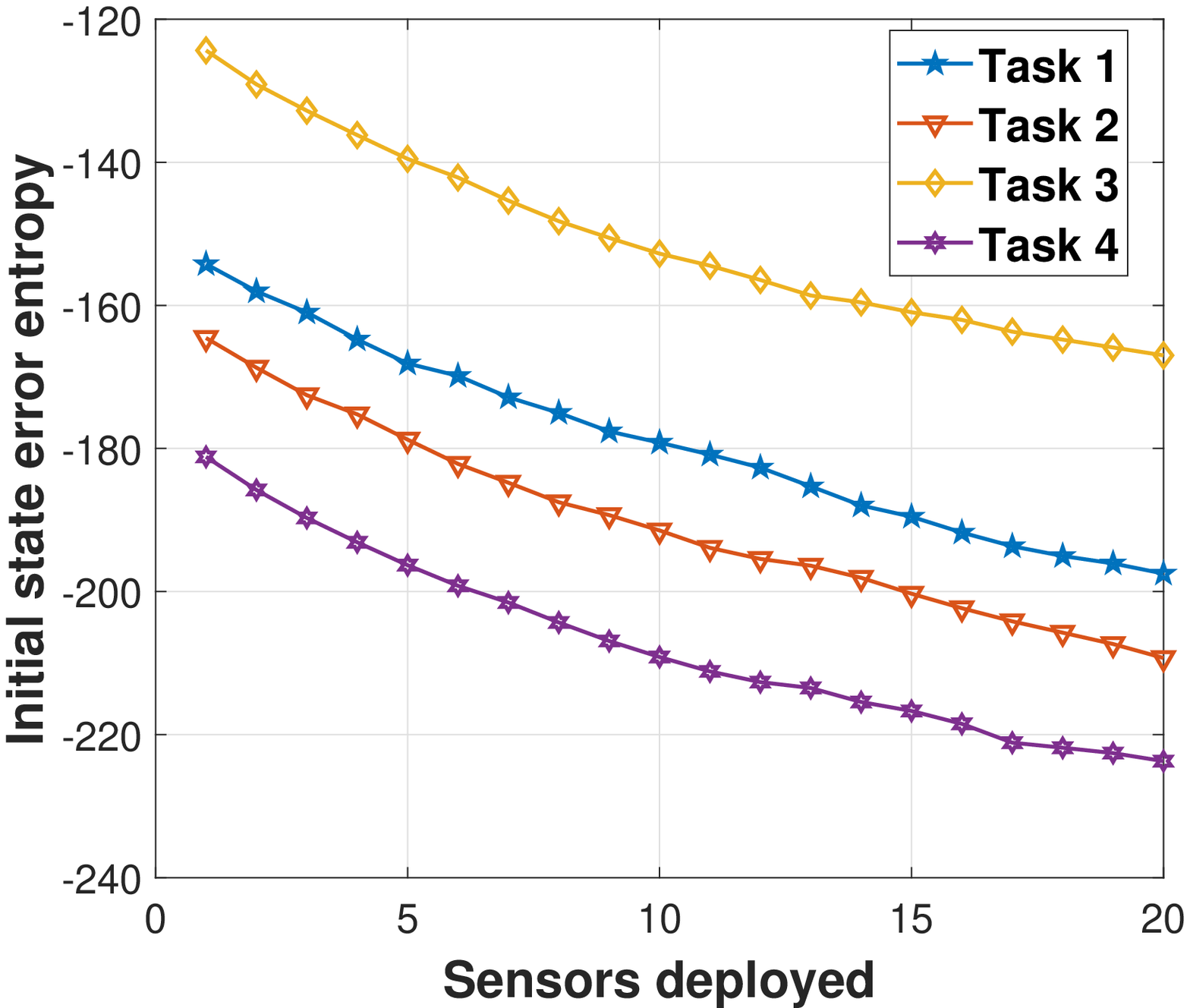}
\caption{}
\end{subfigure}

\vspace*{15pt}
\begin{subfigure}[b]{\linewidth}
\centering
\includegraphics*[viewport=0 0 510 450, width = 1.5in, height = 1.4in]{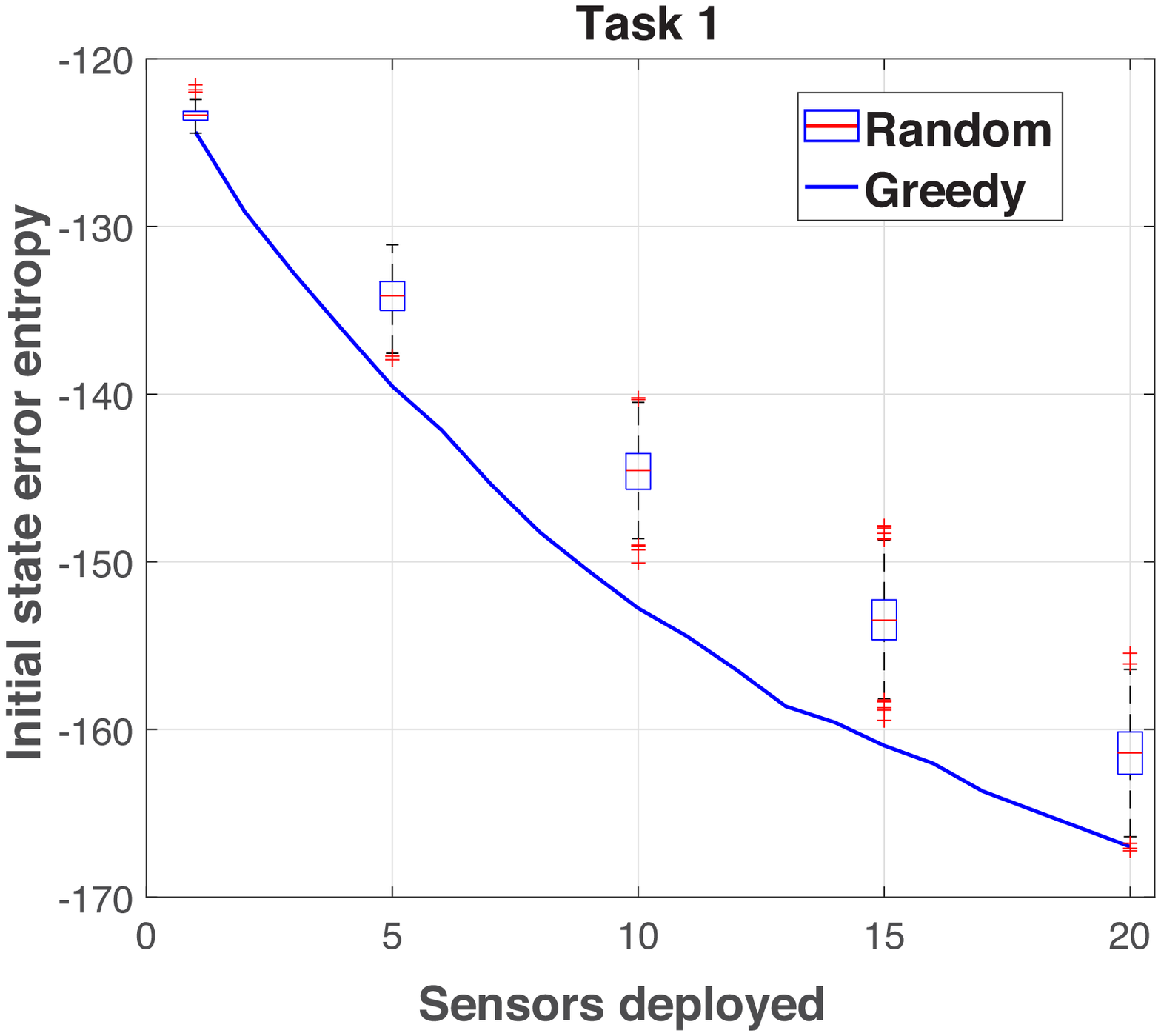}\hfill
\includegraphics*[viewport=0 0 510 450, width = 1.5in, height = 1.4in]{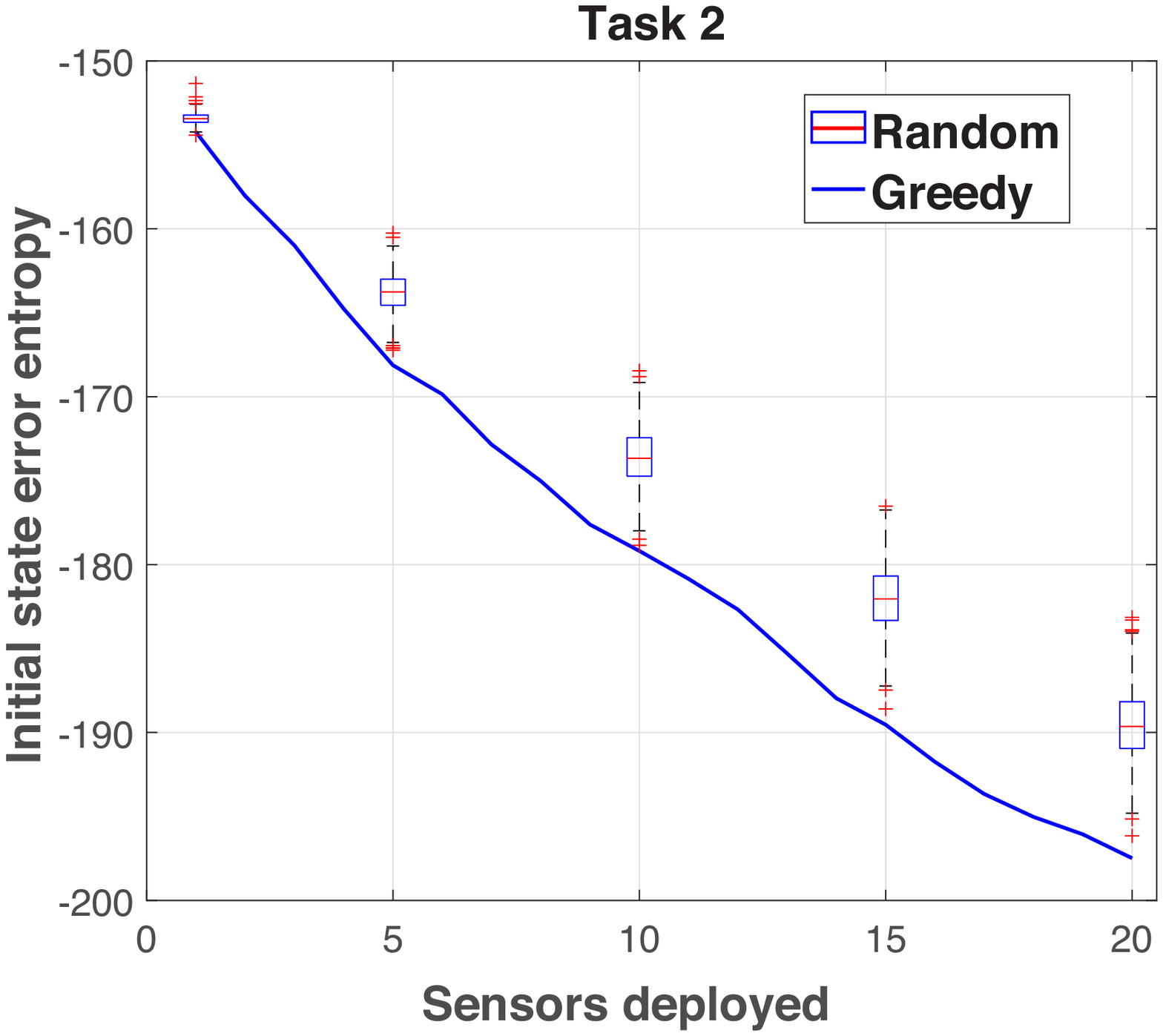}\hfill
\includegraphics*[viewport=0 0 510 450, width = 1.5in, height = 1.4in]{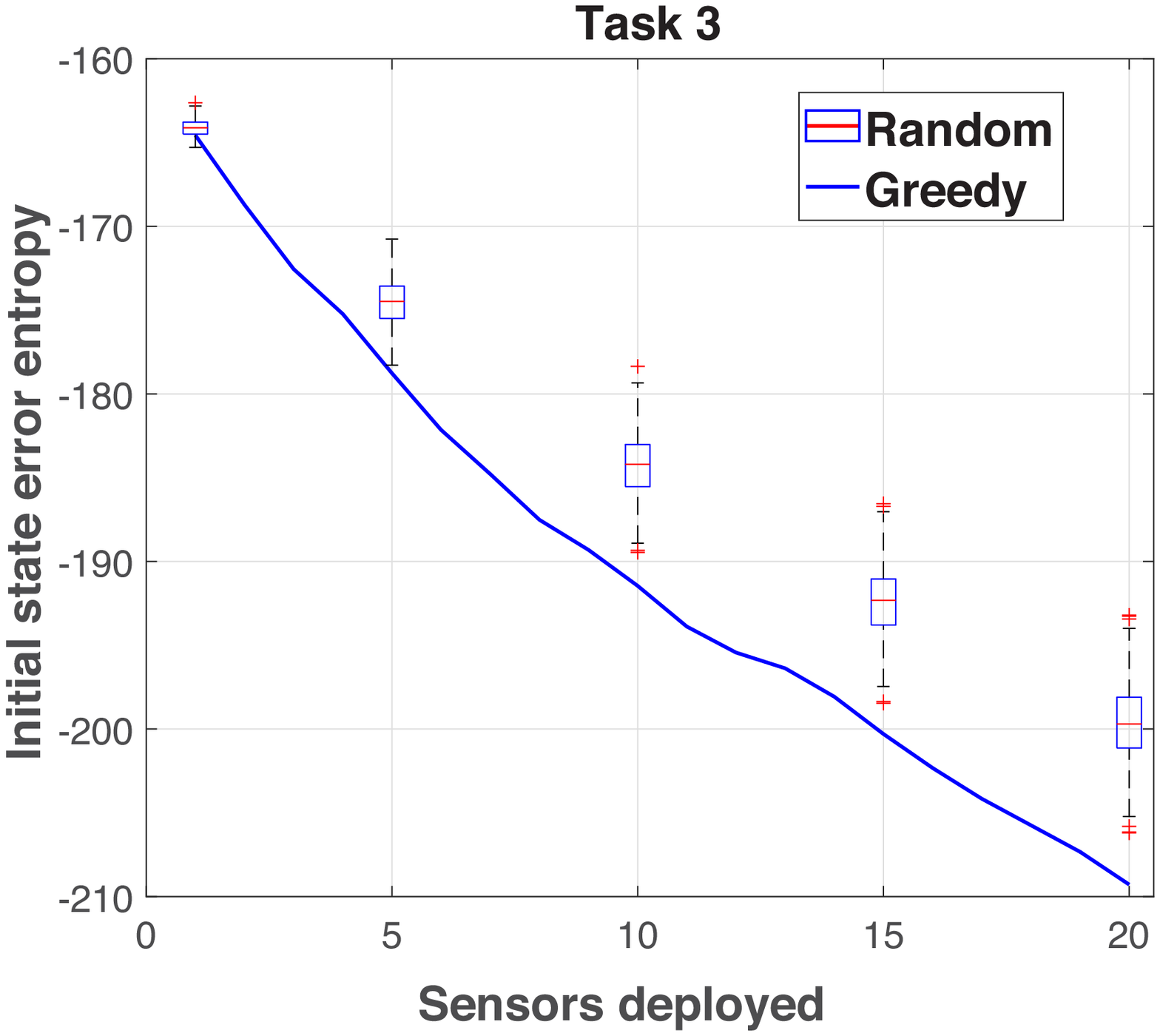}\hfill
\includegraphics*[viewport=0 0 510 450, width = 1.5in, height = 1.4in]{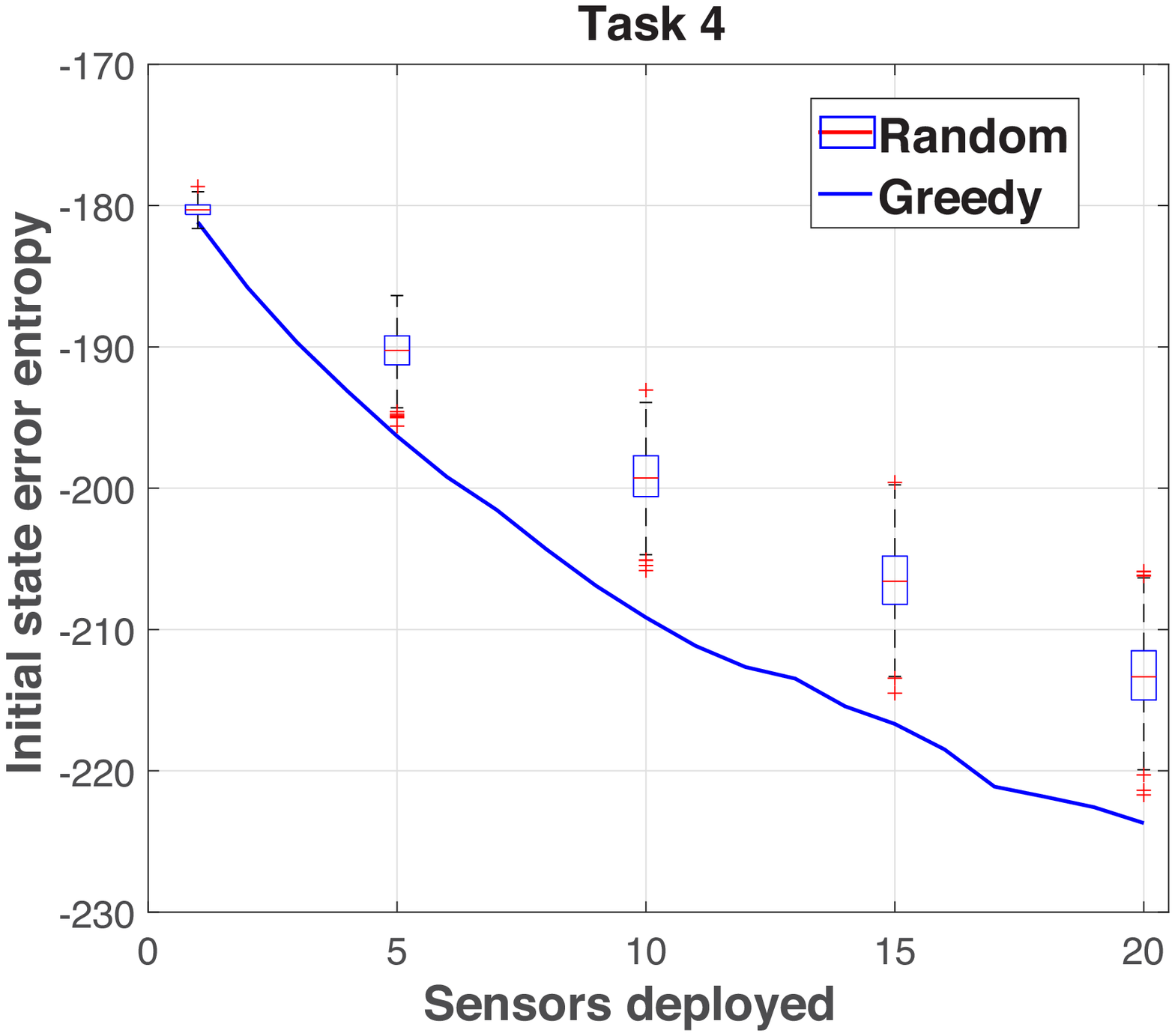}
\caption{}
\end{subfigure}
\caption{Greedy selection with maximization of the minimum eigenvalue, and performance representation with mean square and entropy of the initial state estimation error in (a) and (b) respectively for all $4$ tasks versus various sensor budgets. A comparison of the performance of Greedy and random selection of sensors for all $4$ tasks is presented in (c) with box plots for random selection.}
\label{fig:initStateErr}
\end{figure}

The problem can be stated as follows: selection of appropriate sensors subject to the sensors number constraint such that the estimation error of the initial state is minimized. We will resort to the identified $\delta-$approximate submodular function, the minimum eigenvalue of the observability Gramian (\ref{eqn:minEigGramian}) to select the sensors. The selection of sensors subject to the maximization of the minimum eigenvalue helps in improving the quality of estimation while inversion of the observability matrix is required. The estimation of the initial state is executed via minimum mean squared error criteria \citep{sergio,tzoumas2018TCNS}. The performance of the selection metric,  i.e., minimum eigenvalue is plotted in \figurename\,\ref{fig:initStateErr} for a total of $4$ tasks. For the evaluation purpose, two metrics are explored namely, mean squared error of the initial state estimation and entropy of the error incurred in the initial state estimation. We have omitted the additive constants in the plots of the entropy. It is observed that the minimum eigenvalue metric performs well with the greedy selection algorithm. We also compared in \figurename\,\ref{fig:initStateErr} the performance with respect to the entropy of initial state estimation error while executing greedy and random sensor selection. The performance of the random selection is presented in the form of box plots. We observe a gap between greedy and random selection in all of the $4$ tasks.

\section{Related Work}

In this section, we acknowledge works related to non-submodular functions and approximations. Also, we will see the distinction between our current work.

There has been work done in the domain of approximating set functions. Specifically for submodular functions, \citep{iyer2} discuss functions which are representable as difference of submodular functions, \citep{goemans,goel,svitkina} approximated submodular functions using value oracle model with polynomial queries. \citep{iyer} used the curvature of submodular functions to improve the bounds of \citep{goemans}. \citep{devanur} approximated the submodular functions concerning cut functions of directed/undirected graphs. While these works focus on approximating submodular functions, the closest ones related to approximating non-submodular functions are \citep{horel, hassidim}. The authors have considered deviations from submodularity in the context of noisy oracle model. Similarly, there are other works related to noisy versions \citep{krause3, krause4, kempe}. The fundamental model was the presence of a submodular function regarding value oracle and then its deviation from submodularity upon corruption from additive or multiplicative noise. But as we have seen that in many applications the objective function is provided in the closed form and therefore, the number of queries are not the primary concern. Since marginals of functions are more useful in the optimization algorithms, we have seen that defining approximation in terms of function marginals is more useful than regarding function value itself, for example, the definition of $\epsilon-$approximate submodular in \citep{horel,NIPS2017_6970}.

The other versions of non-submodularity like restricted and shifted are pointed out in \citep{Du}. The restricted submodularity arise when a function is not submodular in general but only when restricted to some subset, while shifted submodularity was defined as $f_{A\cup B}(X) - f_{A}(X) \leq 1$. \citep{borodin} introduced the notion of weak submodular functions which are monotonous and proved the performance bounds. 

For a general non-submodular function, \citep{das} have introduced submodularity ratio (and generalized submodularity ratio \citep{bian}) which is roughly speaking a measure of the deviation from submodularity. Using generalized submodularity ratio and generalized curvature, \citep{bian} derived performance bound for any non-submodular function. The computation of the parameters in the bound is a challenge as they require an exhaustive combinatorial search. Instead, the currently proposed framework of $\delta-$approximation computes the performance bound in linear complexity. Besides, it also helps in establishing the correspondence of non-submodular function to some submodular functions with closed-form bounds.

\section{Conclusion}

In this work, we have proposed that any set function can be modeled as a $\delta-$approximate submodular function. A new interpretation of closeness to submodularity is presented using region of submodularity (ROS) which is a function of the submodularity ratio. This methodology offers fundamental insight into the proximity of a function to submodularity, and it is shown that better performance guarantee bounds for the greedy algorithm can be achieved by carefully choosing a submodular function in the ROS with minimum total curvature as a function of $\delta$. The expression for performance bound is computable in just linear complexity as compared to the existing ones which require a combinatorial search.

Besides, the present results can be copulated with the existing literature on learning submodular functions to bound the approximation errors in the context of query-base setups, and open new venues of research for interpolation of closed-form set functions that often occur in applications such as sensor and/or feature selection. The future work would include sufficient condition for the existence of submodular functions and lower bounds on the minimum curvature in ROS as a function of $\delta$, such that the sub-optimality guarantees are improved.


\nocite{bach, das1, horel, hassidim, bian, iyer, iyer1, feige, xue}

\bibliography{delSMF}

\begin{thebibliography}{41}
\providecommand{\natexlab}[1]{#1}
\providecommand{\url}[1]{\texttt{#1}}
\expandafter\ifx\csname urlstyle\endcsname\relax
  \providecommand{\doi}[1]{doi: #1}\else
  \providecommand{\doi}{doi: \begingroup \urlstyle{rm}\Url}\fi

\bibitem[Bach(2013)]{bach}
Francis Bach.
\newblock Learning with submodular functions: A convex optimization
  perspective, 2013.
\newblock arXiv:1111.6453.

\bibitem[Bian et~al.(2017)Bian, Buhmann, Krause, and Tschiatschek]{bian}
Andrew~An Bian, Joachim~M. Buhmann, Andreas Krause, and Sebastian Tschiatschek.
\newblock Guarantees for greedy maximization of non-submodular functions with
  applications, 2017.
\newblock arXiv:1703.02100.

\bibitem[Borodin et~al.(2014)Borodin, Le, and Ye]{borodin}
Allan Borodin, Dai Tri~Man Le, and Yuli Ye.
\newblock Weakly submodular functions, 2014.
\newblock arXiv:1401.6697.

\bibitem[Buchbinder et~al.(2012)Buchbinder, Feldman, Naor, , and
  Schwartz]{buchbinder}
N.~Buchbinder, M.~Feldman, J.~Naor, , and R.~Schwartz.
\newblock A tight (1/2) linear-time approximation to unconstrained submodular
  maximization.
\newblock In \emph{FOCS}, 2012.

\bibitem[Chen et~al.(2015)Chen, Hassani, Karbasi, , and Krause]{krause3}
Yuxin Chen, S.~Hamed Hassani, Amin Karbasi, , and Andreas Krause.
\newblock Sequential information maximization: When is greedy near-optimal?
\newblock In \emph{COLT}, 2015.

\bibitem[Conforti and Cornuejols(1984)]{conforti}
Michele Conforti and Gerard Cornuejols.
\newblock Submodular set functions, matroids and the greedy algorithm: tight
  worst-case bounds and some generalizations of the {R}ado-{E}dmonds theorem.
\newblock \emph{Discrete Applied Math}, 7\penalty0 (3):\penalty0 251--274,
  1984.

\bibitem[Das and Kempe(2008)]{Das2}
Abhimanyu Das and David Kempe.
\newblock Algorithms for subset selection in linear regression.
\newblock In \emph{Proceedings of the Fortieth Annual ACM Symposium on Theory
  of Computing}, pages 45--54, 2008.

\bibitem[Das and Kempe(2011)]{das}
Abhimanyu Das and David Kempe.
\newblock Submodular meets spectral: Greedy algorithms for subset selection,
  sparse approximation and dictionary selection.
\newblock In \emph{ICML}, 2011.

\bibitem[Das et~al.(2012)Das, Dasgupta, and Kumar]{das1}
Abhimanyu Das, Anirban Dasgupta, and Ravi Kumar.
\newblock Selecting diverse features via spectral regularization.
\newblock In \emph{NIPS}, 2012.

\bibitem[Devanur et~al.(2013)Devanur, Dughmi, Schwartz, Sharma, and
  Singh]{devanur}
Nikhil~R. Devanur, Shaddin Dughmi, Roy Schwartz, Ankit Sharma, and Mohit Singh.
\newblock On the approximation of submodular functions, 2013.
\newblock arXiv:1304.4948.

\bibitem[Du et~al.(2008)Du, Graham, Pardalos, Wan, Wu, and Zhao]{Du}
Ding-Zhu Du, Ronald~L. Graham, Panos~M. Pardalos, Peng-Jun Wan, Weili Wu, and
  Wenbo Zhao.
\newblock Analysis of greedy approximations with nonsubmodular potential
  functions.
\newblock In \emph{Proceedings of the Nineteenth Annual ACM-SIAM Symposium on
  Discrete Algorithms}, SODA '08, pages 167--175. Society for Industrial and
  Applied Mathematics, 2008.

\bibitem[Feige et~al.(2011)Feige, Mirrokni, and Vondrak]{feige}
U.~Feige, V.~S. Mirrokni, and J.~Vondrak.
\newblock Maximizing non-monotone submodular functions.
\newblock \emph{SIAM J. Comput}, 40\penalty0 (4):\penalty0 1133--1153, 2011.

\bibitem[Goel et~al.(2009)Goel, Karande, Tripathi, and Wang]{goel}
G.~Goel, C.~Karande, P.~Tripathi, and L.~Wang.
\newblock Approximability of combinatorial problems with multi-agent submodular
  cost functions.
\newblock In \emph{FOCS}, 2009.

\bibitem[Goemans et~al.(2009)Goemans, Harvey, Iwata, and Mirrokni]{goemans}
M.~Goemans, N.~Harvey, S.~Iwata, and V.~Mirrokni.
\newblock Approximating submodular functions everywhere.
\newblock In \emph{SODA}, pages 535--544, 2009.

\bibitem[Goldberger et~al.(2000)Goldberger, Amaral, Glass, Hausdorff, Ivanov,
  Mark, Mietus, Moody, Peng, and Stanley]{Goldbergere215}
Ary~L. Goldberger, Luis A.~N. Amaral, Leon Glass, Jeffrey~M. Hausdorff,
  Plamen~Ch. Ivanov, Roger~G. Mark, Joseph~E. Mietus, George~B. Moody,
  Chung-Kang Peng, and H.~Eugene Stanley.
\newblock Physiobank, physiotoolkit, and physionet.
\newblock \emph{Circulation}, 101\penalty0 (23):\penalty0 e215--e220, 2000.

\bibitem[Golovin and Krause(2011)]{golovin}
D.~Golovin and A.~Krause.
\newblock Adaptive submodularity: Theory and applications in active learning
  and stochastic optimization.
\newblock \emph{JAIR}, 42:\penalty0 427--486, 2011.

\bibitem[Guillory and Bilmes(2011)]{guillory}
A.~Guillory and J.~Bilmes.
\newblock Simultaneous learning and covering with adversarial noise.
\newblock In \emph{ICML}, 2011.

\bibitem[Gupta et~al.(2018)Gupta, Pequito, and Bogdan]{gauravACC2018}
Gaurav Gupta, S\'ergio Pequito, and Paul Bogdan.
\newblock Dealing with unknown unknowns: Identification and selection of
  minimal sensing for fractional dynamics with unknown inputs.
\newblock \emph{to appear in American Control Conference 2018}, 2018.
\newblock arXiv:1803.04866.

\bibitem[Hassidim and Singer(2016)]{hassidim}
Avinatan Hassidim and Yaron Singer.
\newblock Submodular optimization under noise, 2016.
\newblock {C}oRR, abs/1601.03095.

\bibitem[Hoi et~al.(2006)Hoi, Jin, Zhu, and Lyu]{hoi}
S.~Hoi, R.~Jin, J.~Zhu, and M.~Lyu.
\newblock Batch mode active learning and its application to medical image
  classification.
\newblock In \emph{ICML}, 2006.

\bibitem[Horel and Singer(2016)]{horel}
Thibaut Horel and Yaron Singer.
\newblock Maximization of approximately submodular functions.
\newblock In \emph{NIPS}, 2016.

\bibitem[Iyer and Bilmes(2012)]{iyer2}
R.~K. Iyer and J.~A. Bilmes.
\newblock Algorithms for approximate minimization of the difference between
  submodular functions, with applications.
\newblock In \emph{UAI}, 2012.

\bibitem[Iyer and Bilmes(2013)]{iyer1}
R.~K. Iyer and J.~A. Bilmes.
\newblock Submodular optimization with submodular cover and submodular knapsack
  constraints.
\newblock In \emph{NIPS}, pages 2436--2444, 2013.

\bibitem[Iyer et~al.(2013)Iyer, Jegelka, and Bilmes]{iyer}
R.~K. Iyer, S.~Jegelka, and J.~A. Bilmes.
\newblock Curvature and optimal algorithms for learning and minimizing
  submodular functions.
\newblock In \emph{NIPS}, pages 2742--2750, 2013.

\bibitem[Kempe et~al.(2003)Kempe, Kleinberg, and Tardos]{kempe}
David Kempe, J.~Kleinberg, and E.~Tardos.
\newblock Maximizing the spread of influence through a social network.
\newblock In \emph{KDD}, 2003.

\bibitem[Krause et~al.(2008)Krause, Singh, and Guestrin]{krause1}
A.~Krause, A.~Singh, and C.~Guestrin.
\newblock Near-optimal sensor placements in gaussian processes: Theory,
  efficient algorithms and empirical studies.
\newblock \emph{Journal of Machine Learning Research}, 9:\penalty0 235--284,
  Jun 2008.

\bibitem[Krause and Cevher(2010)]{cevher}
Andreas Krause and Volkan Cevher.
\newblock Submodular dictionary selection for sparse representation.
\newblock In \emph{ICML}, pages 567--574, 2010.

\bibitem[Li et~al.(2017)Li, Chen, Xiaoming~Sun, and
  Jialin~Zhang]{NIPS2017_6970}
Qiang Li, Wei Chen, Institute of~Computing Xiaoming~Sun, and Institute
  of~Computing Jialin~Zhang.
\newblock Influence maximization with $\epsilon-$almost submodular threshold
  functions.
\newblock In \emph{Advances in Neural Information Processing Systems 30}, pages
  3804--3814. 2017.

\bibitem[Nemhauser et~al.(1978)Nemhauser, Wolsey, , and Fisher]{nemhauser}
G.~Nemhauser, L.~Wolsey, , and M.~Fisher.
\newblock An analysis of approximations for maximizing submodular set
  functions.
\newblock \emph{Mathematical Programming}, 14\penalty0 (1):\penalty0 265--294,
  1978.

\bibitem[Pasqualetti et~al.(2014)Pasqualetti, Zampieri, and Bullo]{Pasqualetti}
F.~Pasqualetti, S.~Zampieri, and F.~Bullo.
\newblock Controllability metrics, limitations and algorithms for complex
  networks.
\newblock \emph{IEEE Transactions on Control of Network Systems}, 1\penalty0
  (1):\penalty0 40--52, March 2014.

\bibitem[Pequito et~al.(2017)Pequito, Clark, and Pappas]{sergio}
S.~Pequito, A.~Clark, and G.~J. Pappas.
\newblock Discrete-time fractional-order multiple scenario-based sensor
  selection.
\newblock In \emph{2017 American Control Conference (ACC)}, pages 5488--5493,
  May 2017.

\bibitem[Schalk et~al.(2004)Schalk, McFarland, Hinterberger, Birbaumer, and
  Wolpaw]{schalk}
G.~Schalk, D.~J. McFarland, T.~Hinterberger, N.~Birbaumer, and J.~R. Wolpaw.
\newblock Bci2000: a general-purpose brain-computer interface ({BCI}) system.
\newblock \emph{IEEE Transactions on Biomedical Engineering}, 51\penalty0
  (6):\penalty0 1034--1043, June 2004.

\bibitem[Schnitzler et~al.(2015)Schnitzler, Yu, , and Mannor]{schnitzler}
F.~Schnitzler, J.~Y. Yu, , and S.~Mannor.
\newblock Sensor selection for crowdsensing dynamical systems.
\newblock In \emph{AISTATS}, 2015.

\bibitem[Singla et~al.(2015)Singla, Tschiatschek, and Krause]{krause4}
A.~Singla, S.~Tschiatschek, and A.~Krause.
\newblock Noisy submodular maximization via adaptive sampling with applications
  to crowdsourced image collection summarization, 2015.
\newblock arXiv:1511.07211.

\bibitem[Summers et~al.(2016)Summers, Cortesi, and Lygeros]{Summers2016OnSA}
Tyler~H. Summers, Fabrizio~L. Cortesi, and John Lygeros.
\newblock On submodularity and controllability in complex dynamical networks.
\newblock \emph{IEEE Transactions on Control of Network Systems}, 3:\penalty0
  91--101, 2016.

\bibitem[Sviridenko(2004)]{sviridenko}
M.~Sviridenko.
\newblock A note on maximizing a submodular set function subject to a knapsack
  constraint.
\newblock \emph{Operations Research Letters}, 32\penalty0 (1):\penalty0 41--43,
  2004.

\bibitem[Svitkina and Fleischer(2008)]{svitkina}
Z.~Svitkina and L.~Fleischer.
\newblock Submodular approximation: Sampling-based algorithms and lower bounds.
\newblock In \emph{FOCS}, pages 697--706, 2008.

\bibitem[Tropp(2004)]{tropp}
J.~Tropp.
\newblock Greed is good: Algorithmic results for sparse approximation.
\newblock \emph{IEEE Trans. Information Theory}, 50:\penalty0 2231--2242, 2004.

\bibitem[Tzoumas et~al.(2018)Tzoumas, Xue, Pequito, Bogdan, and
  Pappas]{tzoumas2018TCNS}
V.~Tzoumas, Y.~Xue, S.~Pequito, P.~Bogdan, and G.~J. Pappas.
\newblock Selecting sensors in biological fractional-order systems.
\newblock \emph{IEEE Transactions on Control of Network Systems}, pages 1--1,
  2018.

\bibitem[Xue et~al.(2016)Xue, Pequito, Coelho, Bogdan, and Pappas]{xue}
Yuankun Xue, Sergio Pequito, Joana~R. Coelho, Paul Bogdan, and George~J.
  Pappas.
\newblock Minimum number of sensors to ensure observability of physiological
  systems: a case study.
\newblock In \emph{Allerton}, 2016.

\bibitem[Zhang(2008)]{zhang}
T.~Zhang.
\newblock Adaptive forward-backward greedy algorithm for sparse learning with
  linear models.
\newblock In \emph{NIPS}, 2008.

\end{thebibliography}

\appendix
\onecolumn

The Appendix is organized as follows. In Appendix\,\ref{append:pfThm3}, we provide the detailed proof of Theorem\,\ref{thm:deltaApproxGreedy}. In Appendix\,\ref{append:pfThm2}, we present an alternate version of the Theorem\,\ref{thm:submodCurv} using the introduced definition of greedy curvature in (\ref{eqn:greedyCurv}). Finally, in the Appendix\,\ref{append:pfProp}, we provide proofs of the three Propositions stated in Section\,\ref{sec:Appl}.

\section{Proof of Theorem\,\ref{thm:deltaApproxGreedy}}
\label{append:pfThm3}

\begin{proof} Let $S_{G}$ be the output of the greedy algorithm and therefore $f(\Omega^{\ast}) \geq f(S_{G})$. At the $i$th stage of the algorithm,  if $S_{G}^{i}$ is the selected set then $f(\Omega^{\ast}\,\cup\,S_{G}^{i})$ can be expanded in two ways. First we have
\begin{eqnarray}
f(\Omega^{\ast}\cup S_{G}^{i}) &=& f(\Omega^{\ast}) + f_{\Omega^{\ast}}(S_{G}^{i}) = f(\Omega^{\ast}) + \sum\limits_{s_{j}\in S_{G}^{i}\setminus \Omega^{\ast}}f_{\Omega^{\ast}\cup S_{G}^{j-1}}(s_{j})\nonumber\\
&\geq& f(\Omega^{\ast}) + (1-\alpha_{G})\sum\limits_{s_{j}\in S_{G}^{i}\setminus \Omega^{\ast}}f_{S_{G}^{j-1}}(s_{j}),
\label{eqn:midEqn1}
\end{eqnarray}
\noindent where the inequality is written using the definition of greedy curvature from (\ref{eqn:greedyCurv}). On the other hand, the expansion is as follows
\begin{equation}
f(\Omega^{\ast}\cup S_{G}^{i}) = f(S_{G}^{i}) + f_{S_{G}^{i}}(\Omega^{\ast}).
\label{eqn:midEqn2}
\end{equation}
After combining (\ref{eqn:midEqn1}) and (\ref{eqn:midEqn2}), we can write that
\begin{eqnarray*}
f(\Omega^{\ast}) + (1- \alpha_{G})\sum\limits_{s_{j}\in S_{G}^{i}\setminus \Omega^{\ast}}f_{S_{G}^{j-1}}(s_{j}) &\leq& \sum\limits_{s_{j}\in S_{G}^{i}}f_{S_{G}^{j-1}}(s_{j}) + f_{S_{G}^{i}}(\Omega^{\ast}), \nonumber
\end{eqnarray*}
\noindent which can be re-written as 

\begin{eqnarray*}
f(\Omega^{\ast}) \leq \alpha_{G}\sum\limits_{s_{j}\in S_{G}^{i}}f_{S_{G}^{j-1}}(s_{j}) + (1-\alpha_{G})\sum\limits_{s_{j}\in S_{G}^{i}\cap\Omega^{\ast}}f_{S_{G}^{j-1}}(s_{j}) + f_{S_{G}^{i}}(\Omega^{\ast}).\nonumber\\
\end{eqnarray*}

At this point, it should be noted that out of $\text{ROS}(f,\delta)$, we have chosen $g$ which has total curvature of $\alpha_{\delta}$. We can upper bound the last term of the above inequality using (\ref{eqn:altDeltaApprox}) and use the diminishing returns property of submodular function $g$ to write that 
\begin{eqnarray*}
f(\Omega^{\ast})&\leq& \alpha_{G}f(S_{G}^{i}) + (1-\alpha_{G})\sum\limits_{s_{j}\in S_{G}^{i}\cap\Omega^{\ast}}f_{S_{G}^{j-1}}(s_{j}) + (1+\delta)\sum\limits_{\omega\in\Omega^{\ast}\setminus S_{G}^{i}}g_{S_{G}^{i}}(\omega)\nonumber\\
&\leq& \alpha_{G}f(S_{G}^{i}) + (1-\alpha_{G})\sum\limits_{s_{j}\in S_{G}^{i}\cap\Omega^{\ast}}f_{S_{G}^{j-1}}(s_{j}) + \frac{1+\delta}{1-\delta}\sum\limits_{\omega\in\Omega^{\ast}\setminus S_{G}^{i}}f_{S_{G}^{i}}(\omega).\nonumber\\
\end{eqnarray*}

The greedy algorithm at the $(i+1)$th step would select $s_{i+1}$ according to the following 
\begin{equation}
A(i+1) = \max\limits_{a \in \Omega\backslash S_{G}^{i}}f_{S_{G}^{i}}(a) = f_{S_{G}^{i}}(s_{i+1}),
\label{eqn:greedyNextStep}
\end{equation}
where $A(i+1)$ is the gain at the $(i+1)$th step. The last term can be upper bounded by $A(i+1)$ and let us denote the size of set $S_{G}^{i}\,\cap\,\Omega^{\ast}$ as $t_{i}$. Subsequently, we can write that

\begin{eqnarray*}
f(\Omega^{\ast})&\leq& \alpha_{G}f(S_{G}^{i}) + (1-\alpha_{G})\sum\limits_{s_{j}\in S_{G}^{i}\cap\Omega^{\ast}}f_{S_{G}^{j-1}}(s_{j}) + \frac{1+\delta}{1-\delta}(k-t_{i})f_{S_{G}^{i}}(s_{i+1})\nonumber\\
&\leq& \alpha_{G}f(S_{G}^{i}) + \sum\limits_{s_{j}\in S_{G}^{i}\cap\Omega^{\ast}}\left\{(1 - \alpha_{G})f_{S_{G}^{j-1}}(s_{j}) - f_{S_{G}^{i}}(s_{i+1})\right\} + k\frac{1+\delta}{1-\delta}f_{S_{G}^{i}}(s_{i+1}).\nonumber\\
\end{eqnarray*}

\noindent where we have used the fact that $\delta$ is feasible according to Lemma\,\ref{lemm:necessaryCond} and hence $\frac{1+\delta}{1-\delta}\geq\frac{1}{\gamma_{f}}\geq 1$. The summation term in the above inequality can be upper-bounded by $0$ using the definition of greedy curvature from (\ref{eqn:greedyCurv}) and we obtain
\begin{eqnarray}
f(\Omega^{\ast})&\leq& \alpha_{G}f(S_{G}^{i}) + k\frac{1+\delta}{1-\delta}f_{S_{G}^{i}}(s_{i+1}).
\label{eqn:finalForm1}
\end{eqnarray}

We shall now upperbound the greedy curvature in terms of $\alpha_{\delta}$. In that process, we can write that

\begin{eqnarray}
\min\limits_{a \in S_{G}\setminus (S_{G}^{i-1}\cup \Omega^{\ast})}\frac{f_{S_{G}^{i-1}\cup \Omega^{\ast}}(a)}{f_{S_{G}^{i-1}}(a)} &\geq& \frac{1-\delta}{1+\delta}\min\limits_{a \in S_{G}\setminus (S_{G}^{i-1}\cup \Omega^{\ast})}\frac{g_{S_{G}^{i-1}\cup \Omega^{\ast}}(a)}{g_{S_{G}^{i-1}}(a)} \nonumber\\
&\geq&\frac{1-\delta}{1+\delta}\min\limits_{a \in \Omega}\frac{g_{\Omega\setminus a}(a)}{g(a)} = \frac{1-\delta}{1+\delta}(1-\alpha_{\delta}),
\label{eqn:greedyCurvBnd1}
\end{eqnarray}
\noindent where the first inequality is written using (\ref{eqn:altDeltaApprox}) and second inequality is using the property of submodular functions. The last equality is using the definition of total curvature from (\ref{eqn:totalCurv}). Using the similar approach, we can again write that
\begin{eqnarray}
\min\limits_{\substack{a \in (S_{G}\cap\Omega^{\ast})\setminus S_{G}^{i-1} \\ i\leq j \leq k}}\frac{f_{S_{G}^{j-1}}(s_{j})}{f_{S_{G}^{i-1}}(a)} &\geq& \min\limits_{\substack{a \in (S_{G}\cap\Omega^{\ast})\setminus S_{G}^{i-1} \\ i\leq j \leq k}}\frac{f_{S_{G}^{j-1}}(a)}{f_{S_{G}^{i-1}}(a)}\nonumber\\
&\geq&\frac{1-\delta}{1+\delta}\min\limits_{\substack{a \in (S_{G}\cap\Omega^{\ast})\setminus S_{G}^{i-1} \\ i\leq j \leq k}}\frac{g_{S_{G}^{j-1}}(a)}{g_{S_{G}^{i-1}}(a)} \nonumber\\
&\geq& \frac{1-\delta}{1+\delta}\min\limits_{\substack{a \in (S_{G}\cap\Omega^{\ast})\setminus S_{G}^{i-1} \\ i\leq j \leq k}}\frac{g_{\Omega\setminus a}(a)}{g(a)}
\nonumber\\
&\geq& \frac{1-\delta}{1+\delta}\min\limits_{a\in\Omega}\frac{g_{\Omega\setminus a}(a)}{g(a)} = \frac{1-\delta}{1+\delta}(1-\alpha_{\delta}).
\label{eqn:greedyCurvBnd2}
\end{eqnarray}
Using the equation (\ref{eqn:greedyCurvBnd1}), (\ref{eqn:greedyCurvBnd2}) and (\ref{eqn:greedyCurv}), we can now conclude that $\alpha_{G}\leq \frac{2\delta}{1+\delta} + \frac{1-\delta}{1+\delta}\alpha_{\delta}$. Therefore, equation (\ref{eqn:finalForm1}) can now be written as

\begin{eqnarray}
f(\Omega^{\ast})&\leq& \left(\frac{2\delta}{1+\delta} + \frac{1-\delta}{1+\delta}\alpha_{\delta} \right)f(S_{G}^{i}) + k\frac{1+\delta}{1-\delta}f_{S_{G}^{i}}(s_{i+1}).
\label{eqn:finalForm}
\end{eqnarray}

The equation (\ref{eqn:finalForm}) is of the form of $\lambda u_{i+1} \geq c - \sum\nolimits_{j=1}^{i}u_{j}$, which has the solution of the form $\sum\nolimits_{i=1}^{k}u_{i} \geq c(1 - (1 - \frac{1}{\lambda})^{k})$ using simple mathematical induction. Therefore, we can write
\begin{eqnarray}
f(S_{G}) &=& \sum\limits_{i=1}^{k}A(i)\nonumber\\
&\geq& \frac{1}{\frac{2\delta}{1+\delta} + \frac{1-\delta}{1+\delta}\alpha_{\delta}}\left(1 - \left(1 - \frac{1}{k}\left(\frac{2\delta}{1+\delta} + \frac{1-\delta}{1+\delta}\alpha_{\delta}\right)\left(\frac{1-\delta}{1+\delta}\right)\right)^{k}\right)OPT\nonumber\\
&\geq& \frac{1}{\frac{2\delta}{1+\delta} + \frac{1-\delta}{1+\delta}\alpha_{\delta}}\left(1 - e^{-\left(\frac{2\delta}{1+\delta} + \frac{1-\delta}{1+\delta}\alpha_{\delta}\right)\frac{1-\delta}{1+\delta}}\right)OPT.
\label{eqn:finalPerfBnd}
\end{eqnarray}
\end{proof}

\section{Alternate proof of Theorem\,\ref{thm:submodCurv}}
\label{append:pfThm2}

\begin{proof}

The reader can benefit from this alternate proof using the new definition of greedy curvature in (\ref{eqn:greedyCurv}). It is shown that the proof can be done without going into the direction of constructing LPs. Now, following the same steps and notations of the proof of Theorem \ref{thm:curvature}, we can write at the $i$th step that

\begin{eqnarray*}
f(\Omega^{\ast}) &\leq& (\alpha_{G}-1)\sum\limits_{s_{j}\in S_{G}^{i}\setminus \Omega^{\ast}}f_{S_{G}^{j-1}}(s_{j}) + \sum\limits_{s_{j}\in S_{G}^{i}}f_{S_{G}^{j-1}}(s_{j}) + f_{S_{G}^{i}}(\Omega^{\ast}) \nonumber\\
&=& \alpha_{G}\sum\limits_{s_{j}\in S_{G}^{i}\setminus \Omega^{\ast}}f_{S_{G}^{j-1}}(s_{j}) + \sum\limits_{s_{j}\in S_{G}^{i}\cap \Omega^{\ast}}f_{S_{G}^{j-1}}(s_{j}) + f_{S_{G}^{i}}(\Omega^{\ast}).\nonumber\\
\end{eqnarray*}
Using the definition of submodularity ratio from (\ref{eqn:genSubMod}) we can upperbound the last term to write that
\begin{eqnarray*}
f(\Omega^{\ast}) &\leq& \alpha_{G}\sum\limits_{s_{j}\in S_{G}^{i}\setminus \Omega^{\ast}}f_{S_{G}^{j-1}}(s_{j}) + \sum\limits_{s_{j}\in S_{G}^{i}\cap \Omega^{\ast}}f_{S_{G}^{j-1}}(s_{j}) + \frac{1}{\gamma}\sum\limits_{\omega \in \Omega^{\ast}\setminus S_{G}^{i}}f_{S_{G}^{i}}(\omega).\nonumber\\
\end{eqnarray*}
Using the property of greedy algorithm from (\ref{eqn:greedyCurv}), the last term would be upper-bounded by $A(i+1) = f_{S_{G}^{i}}(s_{i+1})$ and we can denote the size of the set $S_{G}^{i}\,\cap\,\Omega^{\ast}$ as $t_{i}$ to further write
\begin{eqnarray*}
f(\Omega^{\ast})&\leq& \alpha_{G}f(S_{G}^{i}) + \sum\limits_{s_{j}\in S_{G}^{i}\cap \Omega^{\ast}}(1-\alpha_{G})f_{S_{G}^{j-1}}(s_{j}) + \frac{k-t_{i}}{\gamma} f_{S_{G}^{i}}(s_{i+1})\nonumber\\
&=& \alpha_{G}f(S_{G}^{i}) + \sum\limits_{s_{j}\in S_{G}^{i}\cap \Omega^{\ast}}\left\{(1-\alpha_{G})f_{S_{G}^{j-1}}(s_{j}) - \frac{1}{\gamma}f_{S_{G}^{i}}(s_{i+1})\right\} + \frac{k}{\gamma} f_{S_{G}^{i}}(s_{i+1})\nonumber\\
&\leq& \alpha_{G}f(S_{G}^{i}) + \sum\limits_{s_{j}\in S_{G}^{i}\cap \Omega^{\ast}}\left\{(1-\alpha_{G})f_{S_{G}^{j-1}}(s_{j}) - f_{S_{G}^{i}}(s_{i+1})\right\} + \frac{k}{\gamma} f_{S_{G}^{i}}(s_{i+1}).\nonumber\\
\end{eqnarray*}
where the last inequality is written using the fact that $\gamma \leq 1$. Also, using the definition of greedy curvature from (\ref{eqn:greedyCurv}) the summation term in the above inequality would be upper-bounded by $0$. Therefore, we can say that
\begin{eqnarray}
f(\Omega^{\ast})&\leq& \alpha_{G}f(S_{G}^{i}) + \frac{k}{\gamma} f_{S_{G}^{i}}(s_{i+1})\nonumber\\
&\leq& \alpha f(S_{G}^{i}) + \frac{k}{\gamma} f_{S_{G}^{i}}(s_{i+1})= \alpha\sum\limits_{j=1}^{i}A(j) + \frac{k}{\gamma}\,A(i+1).
\label{eqn:curvSubModProofIter}
\end{eqnarray}
where we have used the fact that $\alpha_{G}\leq \alpha$. Using mathematical induction for (\ref{eqn:curvSubModProofIter}), we can write that
\begin{eqnarray*}
f(S_{G}) &=& \sum\limits_{i=1}^{k}A(i) \geq \frac{1}{\alpha}\left(1 - \left(1 - \frac{\alpha\gamma}{k}\right)^{k}\right)OPT \\
&\geq& \frac{1}{\alpha}\left(1 - e^{-\alpha\gamma}\right)OPT.
\end{eqnarray*} 

\end{proof}

\section{Proofs of the Propositions in Section\,\ref{sec:Appl}}
\label{append:pfProp}

The proofs of the Propositions stating $\delta-$approximation of non-submodular functions are described in this Section. In each proof, with $f$ denoting non-submodular function and $g$ being submodular, we will establish the inequality of the form of
\begin{equation*}
g_{S}(a)\delta_{l}\leq f_{S}(a)\leq \delta_{u}g_{S}(a),
\end{equation*}
to get the expressions for $\delta_{l}$ and $\delta_{u}$.

\subsection{Proposition\,\ref{prop:negTrace}}
\begin{proof}
For the matrix of the form of $W_{S} = \Lambda_{0} + \sum\nolimits_{s\in S}{\bf x}_{s}{\bf x}_{s}^{T}$ or any Gramian, with ordered eigenvalues $\lambda_{n}\geq\hdots\geq\lambda_{2}\geq\lambda_{1}$, we consider a submodular function $g(S) = \text{log}\,\text{det}(W_{S})$. The marginal of $g$ can be written as
\begin{eqnarray}
g_{S}(a) &=& \text{log}\,\text{det}(W_{S\cup\{a\}}) - \text{log}\,\text{det}(W_{S})\nonumber\\
&=& \sum\limits_{i=1}^{n}\text{log}\left(\frac{\lambda_{i}(W_{S\cup\{a\}})}{\lambda_{i}(W_{S})}\right).
\label{eqn:logDetMarg}
\end{eqnarray}

The marginal of the considered non-submodular function, $f(S) = -\text{tr}(W_{S}^{-1})$ can be written as
\begin{eqnarray}
f_{S}(a) &=& -\text{tr}(W_{S\cup\{a\}}^{-1}) + \text{tr}(W_{S}^{-1})\nonumber\\
&=& \sum\limits_{i=1}^{n}-\frac{1}{\lambda_{i}(W_{S\cup\{a\}})} + \frac{1}{\lambda_{i}(W_{S})} = \sum\limits_{i=1}^{n}\frac{\lambda_{i}(W_{S\cup\{a\}}) - \lambda_{i}(W_{S})}{\lambda_{i}(W_{S\cup\{a\}})\lambda_{i}(W_{S})}.
\label{eqn:margnegTr}
\end{eqnarray}

The expression in (\ref{eqn:margnegTr}) can be upper bounded using the relation $1-\frac{1}{x}\leq\text{log}(x), x>0$ as
\begin{eqnarray}
f_{S}(a) \leq \sum\limits_{i=1}^{n}\frac{1}{\lambda_{i}(W_{S})}\text{log}\left(\frac{\lambda_{i}(W_{S\cup\{a\}})}{\lambda_{i}(W_{S})}\right)&\leq&\frac{1}{\lambda_{1}(W_{\phi})}\sum\limits_{i=1}^{n}\text{log}\left(\frac{\lambda_{i}(W_{S\cup\{a\}})}{\lambda_{i}(W_{S})}\right)\nonumber\\
&=&\frac{1}{\lambda_{1}(W_{\phi})}g_{S}(a),
\end{eqnarray}
\noindent where we have used (\ref{eqn:logDetMarg}) in the last equality. The marginal of $f$ can be lower bounded using the relation $x-1 \geq \text{log}(x), x>0$ as follows.
\begin{eqnarray*}
f_{S}(a) \geq \sum\limits_{i=1}^{n}\frac{1}{\lambda_{i}(W_{S\cup\{a\}})}\text{log}\left(\frac{\lambda_{i}(W_{S\cup\{a\}})}{\lambda_{i}(W_{S})}\right)&\geq&\frac{1}{\lambda_{n}(W_{\Omega})}\sum\limits_{i=1}^{n}\text{log}\left(\frac{\lambda_{i}(W_{S\cup\{a\}})}{\lambda_{i}(W_{S})}\right)\\
&=&\frac{1}{\lambda_{n}(W_{\Omega})}g_{S}(a).
\end{eqnarray*}
\end{proof}

\subsection{Proposition\,\ref{prop:minEig1}}

\begin{proof}

For the matrix of the form of $W_{S} = \Lambda_{0} + \sum\nolimits_{s\in S}x_{s}x_{s}^{T}$ or any Gramian, we denote the ordered eigenvalues of $X_{S}$ as $\lambda_{n}\geq\hdots\geq\lambda_{2}\geq\lambda_{1}$. Therefore, the Weyl's inequality for matrices can be written as $\lambda_{i}(W_{S}) + \lambda_{1}(W_{\{a\}})\leq\lambda_{i}(W_{S\cup\{a\}})\leq\lambda_{i}(W_{S}) + \lambda_{n}(W_{\{a\}})$. Also, we have taken the case of $\Lambda_{0} = \beta^{2}I_{N}$. We consider the submodular (or modular in this case) function as $g_{1}(S) = \text{tr}(W_{S})$. The marginal of $g_{1}$ can be written as $g_{1\,S}(a) = \text{tr}(W_{\{a\}})$. The marginal of the considered non-submodular function, $f(S) = \lambda_{1}(W_{S})$ can be upper bounded as
\begin{eqnarray*}
f_{S}(a) = \lambda_{1}(W_{S\cup\{a\}}) - \lambda_{1}(W_{S}) &\leq& \lambda_{n}(W_{\{a\}})\\
&=&\text{tr}(W_{\{a\}})\left(1 - \frac{\sum\limits_{i=1}^{n-1}\lambda_{i}(W_{\{a\}})}{\sum\limits_{i=1}^{n}\lambda_{i}(W_{\{a\}})}\right)\\
&\leq&\text{tr}(W_{\{a\}})\left(1 - \frac{n-1}{n}\frac{\lambda_{1}(W_{\{a\}})}{\lambda_{n}(W_{\{a\}})}\right)\\
&\leq&\text{tr}(W_{\{a\}})\left(1 - \frac{n-1}{n}\min\limits_{\omega\in\Omega}\frac{\lambda_{1}(W_{\{\omega\}})}{\lambda_{n}(W_{\{\omega\}})}\right)\\
&=&g_{1\,S}(a)\left(1 - \frac{n-1}{n}\min\limits_{\omega\in\Omega}\frac{\lambda_{1}(W_{\{\omega\}})}{\lambda_{n}(W_{\{\omega\}})}\right),
\end{eqnarray*}
\noindent where in the first inequality we have used the Weyl's inequality. The marginal of $f$ can again written to be lower bounded as
\begin{eqnarray*}
f_{S}(a) = \lambda_{1}(W_{S\cup\{a\}}) - \lambda_{1}(W_{S}) &\geq& \lambda_{1}(W_{\{a\}})\\
&\geq& \frac{\lambda_{1}(W_{\{a\}})}{n}\left(1 + \frac{\lambda_{n-1}(W_{\{a\}})}{\lambda_{n}(W_{\{a\}})}+\hdots+\frac{\lambda_{1}(W_{\{a\}})}{\lambda_{n}(W_{\{a\}})}\right)\\
&=&\frac{\lambda_{1}(W_{\{a\}})\text{tr}(W_{\{a\}})}{n\lambda_{n}(W_{\{a\}})}\\
&\geq&\left(\frac{1}{n}\min\limits_{\omega\in\Omega}\frac{\lambda_{1}(W_{\{\omega\}})}{\lambda_{n}(W_{\{\omega\}})} \right)\text{tr}(W_{\{a\}})\\
&=&\left(\frac{1}{n}\min\limits_{\omega\in\Omega}\frac{\lambda_{1}(W_{\{\omega\}})}{\lambda_{n}(W_{\{\omega\}})} \right)g_{1\,S}(a).
\end{eqnarray*}

\end{proof}

\subsection{Proposition\,\ref{prop:minEig2}}

\begin{proof}
Similar to the Proof of the Proposition\,\ref{prop:minEig1}, we consider the matrix of the form of $W_{S} = \Lambda_{0} + \sum\nolimits_{s\in S}x_{s}x_{s}^{T}$ or any Gramian, and we denote the ordered eigenvalues of $X_{S}$ as $\lambda_{n}\geq\hdots\geq\lambda_{2}\geq\lambda_{1}$. The Weyl's inequality for matrices can be written as $\lambda_{i}(W_{S}) + \lambda_{1}(W_{\{a\}})\leq\lambda_{i}(W_{S\cup\{a\}})\leq\lambda_{i}(W_{S}) + \lambda_{n}(W_{\{a\}})$. We consider the submodular function as $g_{2}(S) = \lambda_{n}(W_{S})$. The considered non-submodular function, $f(S) = \lambda_{1}(W_{S})$ has the marginals which can be bounded using the Weyl's inequality as follows:
\begin{equation}
\lambda_{1}(W_{\{a\}})\leq f_{S}(a) = \lambda_{1}(W_{S\cup\{a\}}) - \lambda_{1}(W_{S}) \leq \lambda_{n}(W_{\{a\}}).
\label{eqn:minEigIneq}
\end{equation}
The marginal of the submodular function $g_{2}$ can be lower bounded as follows.
\begin{eqnarray*}
g_{2\,S}(a)=\lambda_{n}(W_{S\cup\{a\}}) - \lambda_{n}(W_{S})&\geq&\lambda_{1}(W_{\{a\}})\\
&=&\frac{\lambda_{1}(W_{\{a\}})}{\lambda_{n}(W_{\{a\}})}\lambda_{n}(W_{\{a\}})\\
&\geq&\left(\min\limits_{\omega\in\Omega}\frac{\lambda_{1}(W_{\{\omega\}})}{\lambda_{n}(W_{\{\omega\}})}\right)\lambda_{n}(W_{\{a\}})\\
&\geq&\left(\min\limits_{\omega\in\Omega}\frac{\lambda_{1}(W_{\{\omega\}})}{\lambda_{n}(W_{\{\omega\}})}\right)f_{S}(a).
\end{eqnarray*}
\noindent where in the last inequality we have used (\ref{eqn:minEigIneq}). Similarly, we can upperbound the marginal of $g_{2}$ as follows.
\begin{eqnarray*}
g_{2\,S}(a)  = \lambda_{n}(W_{S\cup\{a\}}) - \lambda_{n}(W_{S})&\leq&\lambda_{n}(W_{\{a\}})\\
&=& \frac{\lambda_{n}(W_{\{a\}})}{\lambda_{1}(W_{\{a\}})}\lambda_{1}(W_{\{a\}})\\
&\leq&\left(\max\limits_{\omega\in\Omega}\frac{\lambda_{n}(W_{\{\omega\}})}{\lambda_{1}(W_{\{\omega\}})}\right)\lambda_{1}(W_{\{a\}})\\
&\leq&\left(\max\limits_{\omega\in\Omega}\frac{\lambda_{n}(W_{\{\omega\}})}{\lambda_{1}(W_{\{\omega\}})}\right)f_{S}(a).
\end{eqnarray*}
Therefore, we can bound the marginal of $f$ on the both sides as follows:
\begin{equation*}
\frac{1}{\max\limits_{\omega\in\Omega}\frac{\lambda_{n}(W_{\{\omega\}})}{\lambda_{1}(W_{\{\omega\}})}}g_{2\,S}(a)\leq f_{S}(a) \leq \frac{1}{\min\limits_{\omega\in\Omega}\frac{\lambda_{1}(W_{\{\omega\}})}{\lambda_{n}(W_{\{\omega\}})}}g_{2\,S}(a),
\end{equation*}
or,
\begin{equation*}
{\min\limits_{\omega\in\Omega}\frac{\lambda_{1}(W_{\{\omega\}})}{\lambda_{n}(W_{\{\omega\}})}}g_{2\,S}(a)\leq f_{S}(a) \leq {\max\limits_{\omega\in\Omega}\frac{\lambda_{n}(W_{\{\omega\}})}{\lambda_{1}(W_{\{\omega\}})}}g_{2\,S}(a).
\end{equation*}

\end{proof}

\end{document}